\documentclass{article}

\usepackage[nonatbib, final]{neurips_2020}

\usepackage[utf8]{inputenc} 
\usepackage[T1]{fontenc}    
\usepackage{hyperref}       
\usepackage{url}            
\usepackage{booktabs}       
\usepackage{amsfonts, amssymb, amsmath}       
\usepackage{xfrac}       
\usepackage[final]{microtype}      

\usepackage[most]{tcolorbox} 

\usepackage[group-separator={,}]{siunitx}

\usepackage{graphicx}
\usepackage{subcaption}
\graphicspath{{figures/}}
\usepackage{wrapfig} 

\usepackage{etoolbox}
\BeforeBeginEnvironment{wrapfigure}{\setlength{\intextsep}{0pt}}

\usepackage{amsthm}
\newtheorem{lemma}{Lemma}
\newtheorem{proposition}{Proposition}
\newtheorem{assumption}{Assumption}

\usepackage{algorithm}
\usepackage{algpseudocode}
\newfloat{algorithm}{t}{lop} 

\usepackage[capitalise]{cleveref} 

\newcommand{\cN}{\mathcal{N}}
\newcommand{\cO}{\mathcal{O}}
\newcommand{\cP}{\mathcal{P}}
\newcommand{\cR}{\mathcal{R}}
\newcommand{\cS}{\mathcal{S}}
\newcommand{\cA}{\mathcal{A}}

\newcommand{\grad}[1]{\nabla #1} 
\DeclareMathOperator*{\E}{\mathbb{E}}

\usepackage[backend=biber, maxbibnames=99, url=false, doi=false, isbn=false]{biblatex}
\title{Shared Experience Actor-Critic for \\Multi-Agent Reinforcement Learning}

\author{%
  Filippos Christianos \\
  School of Informatics\\
  University of Edinburgh\\
  \texttt{f.christianos@ed.ac.uk} \\
   \And
  Lukas Schäfer \\
  School of Informatics\\
  University of Edinburgh\\
  \texttt{l.schaefer@ed.ac.uk} \\
   \And
   Stefano V. Albrecht \\
   School of Informatics\\
   University of Edinburgh \\
   \texttt{s.albrecht@ed.ac.uk} \\
}

\begin{document}
\maketitle

\begin{abstract}
Exploration in multi-agent reinforcement learning is a challenging problem, especially in environments with sparse rewards. We propose a general method for efficient exploration by sharing experience amongst agents. Our proposed algorithm, called \emph{Shared Experience Actor-Critic} (SEAC), applies experience sharing in an actor-critic framework by combining the gradients of different agents. We evaluate SEAC in a collection of sparse-reward multi-agent environments and find that it consistently outperforms several baselines and state-of-the-art algorithms by learning in fewer steps and converging to higher returns. In some harder environments, experience sharing makes the difference between learning to solve the task and not learning at all.
\end{abstract}

\section{Introduction}



Multi-agent reinforcement learning (MARL) necessitates exploration of the environment dynamics and of the joint action space between agents. This is a difficult problem due to non-stationarity caused by concurrently learning agents~\cite{Papoudakis2019DealingLearning} and the fact that the joint action space typically grows exponentially in the number of agents. The problem is exacerbated in environments with sparse rewards in which most transitions will not yield informative rewards.


\begin{wrapfigure}[13]{r}{0.25\textwidth}
    \centering 
    \includegraphics[width=0.23\textwidth]{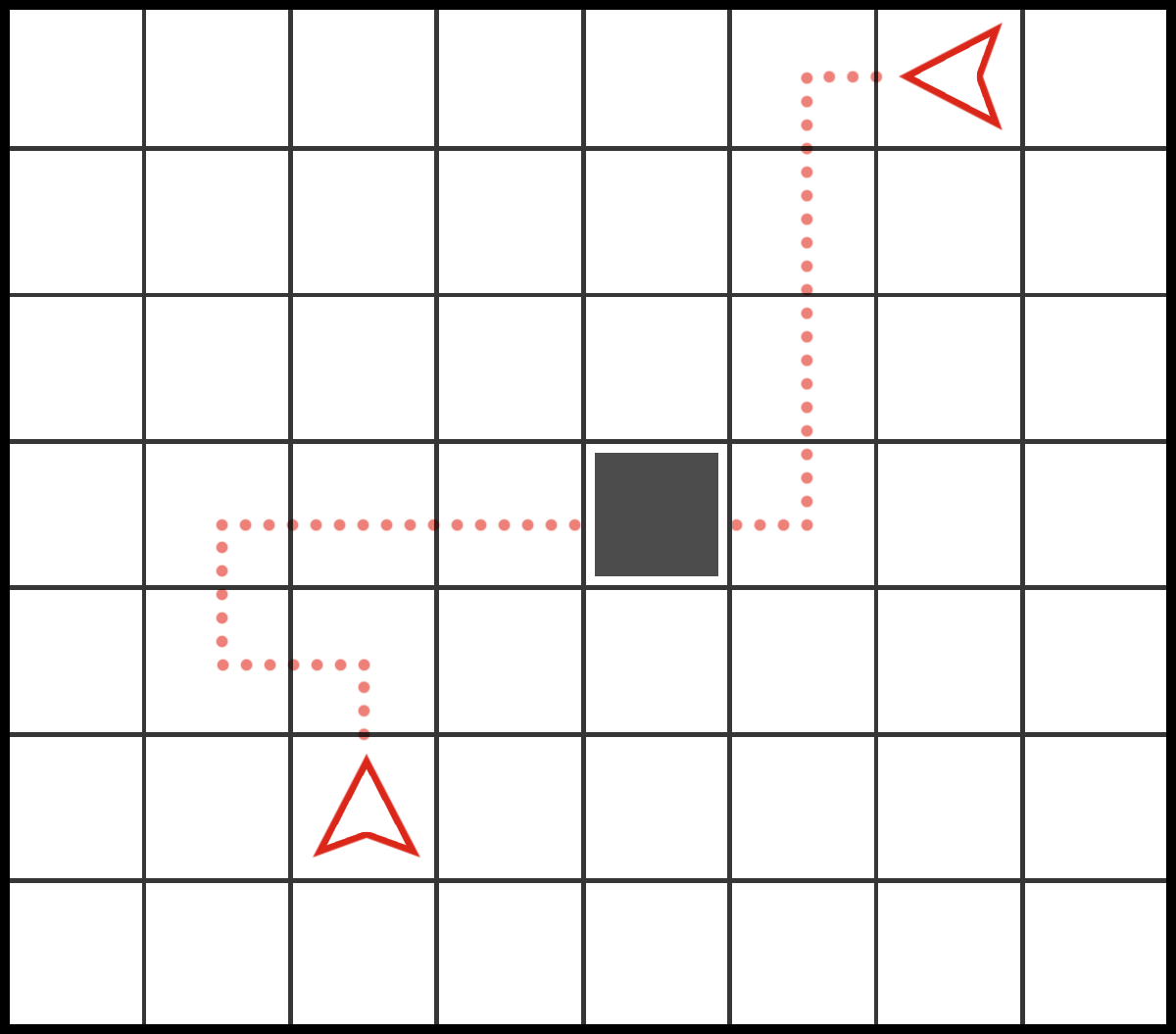}
    \caption{Two randomly-placed agents (triangles) must simultaneously arrive at the goal (square).}
    \label{fig:simple-grid-game}
\end{wrapfigure}%

We propose a general method for efficient MARL exploration by \emph{sharing experience} amongst agents. Consider the simple multi-agent game shown in \Cref{fig:simple-grid-game} in which two agents must simultaneously arrive at a goal. This game presents a difficult exploration problem, requiring the agents to wander for a long period before stumbling upon a reward. When the agents finally succeed, the idea of sharing experience is appealing: both agents can learn how to approach the goal from two different directions after a successful episode by leveraging their collective experience. Such experience sharing facilitates a steady progression of all learning agents, meaning that agents improve at approximately equal rates as opposed to diverging in their learning progress. We show in our experiments that this approach of experience sharing can lead to significantly faster learning and higher final returns.

We demonstrate this idea in a novel actor-critic MARL algorithm, called \emph{Shared Experience Actor-Critic} (SEAC).\footnote{We provide open-source implementations of SEAC in \url{www.github.com/uoe-agents/seac} and our two newly developed multi-agent environments: \url{www.github.com/uoe-agents/lb-foraging} (LBF) and  \url{www.github.com/uoe-agents/robotic-warehouse} (RWARE).} SEAC operates similarly to independent learning~\cite{Tan1993Multi-AgentAgents} but updates the actor and critic parameters of an agent by combining gradients computed on the agent's experience with 
weighted gradients computed on other agents' experiences. We evaluate SEAC in four sparse-reward multi-agent environments and find that it learns substantially faster (up to 70\% fewer required training steps) and achieves higher final returns compared to several baselines, including: independent learning without experience sharing; using data from all agents to train a single shared policy; and MADDPG~\cite{LoweMulti-AgentEnvironments}, QMIX~\cite{Rashid2018QMIX:Learning}, and ROMA~\cite{wang2020roma}. Sharing experience with our implementation of SEAC increased running time by less than 3\% across all environments compared to independent learning.

\section{Related Work}\label{sec:related_work}

\textbf{Centralised Training with Decentralised Execution:} The prevailing MARL paradigm of centralised training with decentralised execution (CTDE)~\cite{Oliehoek2008OptimalPOMDPs, Rashid2018QMIX:Learning, LoweMulti-AgentEnvironments} assumes a training stage during which the learning algorithm can access data from all agents to learn decentralised (locally-executable) agent policies. CTDE algorithms such as MADDPG~\cite{LoweMulti-AgentEnvironments} and COMA~\cite{Foerster2018CounterfactualGradients} learn powerful critic networks conditioned on joint observations and actions of all agents. A crucial difference to SEAC is that algorithms such as MADDPG 
only reinforce an agent's own tried actions, while SEAC uses shared experience to reinforce good actions tried by any agent, without learning the more complex joint-action critics. Our experiments show that MADDPG was unable to learn effective policies in our sparse-reward environments while SEAC learned successfully in most cases.

\textbf{Agents Teaching Agents:} There are approaches to leverage expertise of \textit{teacher} agents to address the issue of sample complexity in training a \textit{learner} agent~\cite{da2020agents}. Such teaching can be regarded as a form of transfer learning~\cite{pan2009survey} among RL agents. The \textit{teacher} would either implicitly or explicitly be asked to evaluate the behaviour of the \textit{learner} and send instructions to the other agent. Contrary to our work, most such approaches do focus on single-agent RL~\cite{clouse1996learning,fachantidis2019learning}. However, even in such teaching approaches for multi-agent systems~\cite{da2017simultaneously,da2018autonomously} experience is shared in the form of knowledge exchange following a teacher-learner protocol. Our approach shares agent trajectories for learning and therefore does not rely on the exchange of explicit queries or instructions, introducing minimal additional cost.

\textbf{Learning from Demonstrations:} Training agents from trajectories~\cite{schaal1997learning} of other agents~\cite{zimmer2014teacher} or humans~\cite{taylor2011integrating} is a common case of teaching agents. 
Demonstration data can be used to derive a policy~\cite{ho2016generative} which might be further refined using typical RL training~\cite{gao2018reinforcement} or to shape the rewards biasing towards previously seen expert demonstrations~\cite{brys2015reinforcement}. These approaches leverage expert trajectories to speed up or simplify learning for single-agent problems. 
In contrast, SEAC makes use of trajectories from other agents which are generated by concurrently learning agents in a multi-agent system. As such, we aim to speed up and synchronise training in MARL whereas learning from demonstrations focuses on using previously generated data for application in domains like robotics where generating experience samples is expensive.

\textbf{Distributed Reinforcement Learning:} Sharing experience among agents is related to recent work in distributed RL. These methods aim to effectively use large-scale computing resources for RL. Asynchronous methods such as A3C~\cite{Mnih2016AsynchronousLearning} execute multiple actors in parallel to generate trajectories more efficiently and break data correlations. Similarly, IMPALA~\cite{Espeholt2018IMPALA:Architectures} and SEED RL~\cite{espeholt2019seed} are off-policy actor-critic algorithms to distribute data collection across many actors with optimisation being executed on a single learner. Network parameters, observations or actions are exchanged after each episode or timestep, respectively, and off-policy correction is applied. 
However, all these approaches only share experience of multiple actors to speed up learning of a single RL agent and by breaking correlations in the data rather than addressing synchronisation and sample efficiency in MARL.


\textbf{Population-play:} Population-based training is another line of research aiming to improve exploration and coordination in MARL by training a population of diverse sets of agents~\cite{wang2018evolving,Leibo2019MalthusianLearning,jaderberg2019human,long2020evolutionary}. \textcite{Leibo2019MalthusianLearning} note the overall benefits on exploration when sets of agents are dynamically evolved and mixed. 
In their work, some agents share policy networks and are trained alongside evolving sets of agents. 
Similarly to \textcite{Leibo2019MalthusianLearning}, we observe benefits on exploration due to agents influencing each other's trajectories. However, SEAC is different from such population-play as it only trains a single set of distinct policies for all agents, thereby avoiding the significant computational cost involved in training multiple sets of agents.

\section{Technical Preliminaries}\label{sec:background}

\textbf{Markov Games:} We consider Markov games for $N$ agents (e.g.~\cite{Littman1994MarkovLearning}) with partial observability (also known as partially observable stochastic games, e.g.~\cite{hansen2004dynamic}) defined by the tuple $(\cN,\cS, \{O^i\}_{i\in \cN}, \{A^i\}_{i\in \cN}, \Omega, \cP, \{\cR^i\}_{i\in \cN})$, with agents $i\in\cN = \{1,\ldots,N\}$, state space $\cS$, joint observation space $\cO = O^1\times\ldots\times O^N$, and joint action space $\cA = A^1\times\ldots\times A^N$. 
Each agent $i$ only perceives local observations $o^i \in O^i$ which depend on the state and applied joint action via observation function $\Omega: \cS \times \cA \mapsto \Delta(\cO)$. The function $\cP: \cS \times \cA \mapsto \Delta(\cS)$ returns a distribution over successor states given a state and a joint action. $\cR^i: \cS \times \cA \times \cS \mapsto \mathbb{R}$ is the reward function for agent $i$ which gives its individual reward $r^i$. 
The objective is to find policies $\pi = (\pi_1, ..., \pi_N)$ for all agents such that the discounted return of each agent $i$, $G^i=\sum^T_{t=0}{\gamma^tr_t^i}$, is maximised with respect to other policies in $\pi$, formally $\forall_i : \pi_i \in \arg\max_{\pi'_i} \mathbb{E}[G^i | \pi'_i,\pi_{-i}]$ where $\pi_{-i} = \pi \setminus \{\pi_i\}$, 
$\gamma$ is the discount factor, and $T$ the total timesteps of an episode.

In this work, we assume $O=O^1=\ldots=O^N$ and $A=A^1=\ldots=A^N$ in line with other recent works in MARL \cite{sunehag2017value,Rashid2018QMIX:Learning,Foerster2018CounterfactualGradients,mahajan2019maven}. However, in contrast to these works we do not require that agents have identical reward functions, as will be discussed in \Cref{sec:seac}.

\textbf{Policy Gradient and Actor-Critic:} Policy Gradient (PG) algorithms are a class of model-free RL algorithms that aim to directly learn a policy $\pi_\phi$ parameterised by $\phi$, that maximises the expected returns. In REINFORCE~\cite{Williams1992SimpleLearning}, the simplest PG algorithm, this is accomplished by following the gradients of the objective $\grad_\phi J(\phi)=\E_\pi\left[G_t\grad_\phi{\ln{\pi_\phi(a_t|s_t)}}\right]$. Notably, the Markov property is not used, allowing the use of PG in partially observable settings. However, REINFORCE suffers from high variance of gradient estimation.
To reduce variance of gradient estimates, actor-critic (AC) algorithms estimate Monte Carlo returns using a value function $V_\pi(s; \theta)$ with parameters $\theta$. In a multi-agent, partially observable setting, the simplest AC algorithm defines a policy loss for agent $i$
\begin{equation}
    \label{eq:policyloss}
    \mathcal{L}(\phi_i) = -\log\pi(a_t^i|o_t^i;\phi_i)(r_t^i + \gamma V(o_{t+1}^i;\theta_i)-V(o_t^i;\theta_i))
\end{equation}
with a value function minimising
\begin{equation}
    \label{eq:valueloss}
    \mathcal{L}(\theta_i) = ||V(o_t^i; \theta_i) - y_i||^2 \text{ \ with \ } y_i = r_t^i + \gamma V(o_{t+1}^i;\theta_i)
\end{equation}

In practice, when $V$ and $\pi$ are parameterised by neural networks, sampling several trajectories in parallel, using $n$-step returns, regularisation, and other modifications can be beneficial~\cite{Mnih2016AsynchronousLearning}. 
To simplify our descriptions, our methods in \Cref{sec:seac} will be described only as extensions of \Cref{eq:policyloss,eq:valueloss}. In our experiments we use a modified AC algorithm as described in \Cref{sec:experiments:algorithm}.


\section{Shared Experience Actor-Critic}\label{sec:seac}



Our goal is to enable more efficient learning by sharing experience among agents. To facilitate experience sharing, we assume environments in which the local policy gradients of agents provide useful learning directions for all agents. Intuitively, this means that agents can learn from the experiences of other agents without necessarily having identical reward functions. Examples of such environments can be found in \Cref{sec:experiments}.

In each episode, each agent generates one on-policy trajectory.
Usually, when on-policy training is used, RL algorithms only use the experience of each agent's own sampled trajectory to update the agent's networks with respect to \Cref{eq:policyloss}. Here, we propose to also use trajectories of other agents while considering that it is \emph{off-policy} data, i.e.\ the trajectories are generated by agents executing different policies than the one optimised. Correcting for off-policy samples requires importance sampling. The loss for such off-policy policy gradient optimisation from a behavioural policy $\beta$ can be written as
\begin{equation}
    \label{eq:offpolicy.policyloss}
    \grad_{\phi}\mathcal{L}(\phi) = -\frac{\pi(a_t|o_t; \phi)}{\beta(a_t|o_t)}\grad_\phi\log\pi(a_t|o_t;\phi)(r_t + \gamma V(o_{t+1};\theta)-V(o_t;\theta))
\end{equation}
In the AC framework of \Cref{sec:background}, we can extend the policy loss to use the agent's own trajectories (denoted with $i$) along with the experience of other agents (denoted with $k$), shown below: 
\begin{align}
    \begin{split}
        \label{eq:policyloss_shared}
        \mathcal{L}(\phi_i) = &-\log\pi(a_t^i|o_t^i;\phi_i)(r_t^i + \gamma V(o_{t+1}^i;\theta_i)-V(o_t^i;\theta_i)) \\
        &- \lambda\sum_{k\neq i}{\frac{\pi(a_t^k|o_t^k; \phi_i)}{\pi(a_t^k|o_t^k; \phi_k)}\log\pi(a_t^k|o_t^k; \phi_i)(r_t^k + \gamma V(o_{t+1}^k;\theta_i)-V(o_t^k;\theta_i)) }
    \end{split}
\end{align}
Using this loss function, each agent is trained on both on-policy data while also using the off-policy data collected by all other agents at each training step. The value loss, in a similar fashion, becomes
\begin{equation}
    \begin{aligned}
    \label{eq:valueloss_shared}
        \mathcal{L}(\theta_i) &= ||V(o_t^i; \theta_i) - y_i^i||^2 + \lambda\sum_{k\neq i}\frac{\pi(a_t^k|o_t^k; \phi_i)}{\pi(a_t^k|o_t^k; \phi_k)}||V(o_t^k; \theta_i) - y_k^i||^2 \\
        y_k^i &= r_t^k + \gamma V(o_{t+1}^k;\theta_i)
    \end{aligned}
\end{equation}

We show how to derive the losses in \Cref{eq:policyloss_shared,eq:valueloss_shared} for the case of two agents in \Cref{appendix:loss_derivation} (generalisation to more agents is possible). The hyperparameter $\lambda$ weights the experience of other agents; we found SEAC to be largely insensitive to values of $\lambda$ and use $\lambda=1$ in our experiments. A sensitivity analysis can be found in \Cref{sec:additional_details}. We refer to the resulting algorithm as \emph{Shared Experience Actor-Critic} (SEAC) and provide pseudocode in \Cref{alg:SEAC}.

Due to the random weight initialisation of neural networks, each agent is trained from experience generated from different policies, leading to more diverse exploration. Similar techniques, such as annealing $\epsilon$-greedy policies to different values of $\epsilon$, have been observed~\cite{Mnih2016AsynchronousLearning} to improve the performance of algorithms.

\begin{algorithm}[t]
    \caption{Shared Experience Actor-Critic Framework}
    \label{alg:SEAC}
    \begin{algorithmic}
        \For{timestep $t=1\dots$}
            \State Observe $o^1_t\dots o^N_t$
            \State Sample actions $a_t^1, \dots, a_t^N$ from $P(o^1_t;\phi_1),\dots,P(o^N_t;\phi_N)$
            \State Execute actions and observe $r_t^1,\dots,r_t^N$ and $o_{t+1}^1,\dots,o_{t+1}^N$
            \For {agent $i=1\dots N$}
                \State Perform gradient step on $\phi_i$ by minimising \cref{eq:policyloss_shared}
                \State Perform gradient step on $\theta_i$ by minimising \cref{eq:valueloss_shared}
            \EndFor
        \EndFor
    \end{algorithmic}
\end{algorithm}

It is possible to apply a similar concept of experience sharing to off-policy deep RL methods such as DQN~\cite{mnih2015human}. We provide a description of experience sharing with DQN in \Cref{sec:appendix_seql}. Since DQN is an off-policy algorithm, experience generated by different policies can be used for optimisation without further considerations such as importance sampling. However, we find deep off-policy methods to exhibit rather unstable learning~\cite{vanHasselt2018DeepTriad} compared to on-policy AC. We consider the generality of our method a strength, and believe it can improve other multi-agent algorithms (e.g. AC with centralised value function).

\section{Experiments}\label{sec:experiments}

We conduct experiments on four sparse-reward multi-agent environments and compare SEAC to two baselines as well as three state-of-the-art MARL algorithms: MADDPG~\cite{LoweMulti-AgentEnvironments}, QMIX~\cite{Rashid2018QMIX:Learning} and ROMA~\cite{wang2020roma}.

\subsection{Environments}
\label{sec:experiments:environments}

The following multi-agent environments were used in our evaluation. More detailed descriptions of these environments can be found in \Cref{apdx:envs}.

\textbf{Predator Prey (PP), \cref{fig:envs:predatorprey}:} First, we use the popular PP environment adapted from the Multi-agent Particle Environment framework~\cite{LoweMulti-AgentEnvironments}. In our sparse-reward variant, three predator agents must catch a prey by coordinating and approaching it simultaneously. The prey is a slowly moving agent that was pretrained with MADDPG and dense rewards to avoid predators. If at least two predators are adjacent to the prey, then they succeed and each receive a reward of one. Agents are penalised for leaving the bounds of the map, but otherwise receive zero reward.

\textbf{Starcraft Multi-Agent Challenge (SMAC), \cref{fig:envs:smac}:} The SMAC~\cite{Samvelyan2019TheChallenge} environment was used in several recent MARL works~\cite{Rashid2018QMIX:Learning, Foerster2018CounterfactualGradients,wang2020roma}. SMAC originally uses dense reward signals and is primarily designed to test solutions to the multi-agent credit assignment problem. We present experiments on a simple variant that uses sparse rewards. In this environment, agents have to control a team of marines each represented by a single agent, to fight against an equivalent team of marines controlled by the game AI. With sparse rewards agents receive a single non-zero reward at the final timestep of each episode: a victory rewards $1$, while a defeat $-1$.

\textbf{Level-Based Foraging (LBF), \cref{fig:envs:foraging}:} LBF~\cite{Albrecht2013ASystems,as2017aamas} is a mixed cooperative-competitive game which focuses on the coordination of the agents involved. Agents of different skill levels navigate a grid world and collect foods by cooperating with other agents if required. Four tasks of this game will be tested, with a varied number of agents, foods, and grid size. Also, a cooperative variant will be tested. The reported returns are the fraction of items collected in every episode.

\textbf{Multi-Robot Warehouse (RWARE), \cref{fig:envs:rware}:}  This multi-agent environment (similar to the one used in~\cite{albrecht2016exploiting}) simulates robots that move goods around a warehouse, similarly to existing real-world applications~\cite{Wurman2008}. The environment requires agents (circles) to move requested shelves (coloured squares) to the goal posts (letter `G') and back to an empty location. It is a partially-observable collaborative environment with a very sparse reward signal, since agents have a limited view area and are rewarded only upon successful delivery. In the results, we report the total returns given by the number of deliveries over an episode of $500$ timesteps on four different tasks in this environment.

\begin{figure}[t]
    \subfigure[t]{0.24\textwidth}%
        \includegraphics[width=\linewidth]{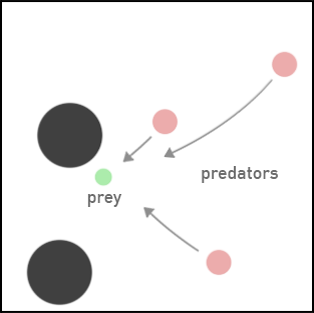}
        \caption{Predator Prey}\label{fig:envs:predatorprey}
    \endsubfigure\hfill
    \subfigure[t]{0.24\textwidth}
        \includegraphics[width=\linewidth]{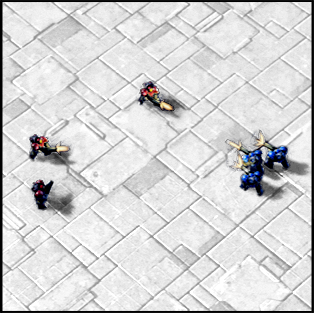}
        \caption{SMAC - 3 marines}\label{fig:envs:smac}
    \endsubfigure\hfill
    \subfigure[t]{0.24\textwidth}
        \includegraphics[width=\linewidth]{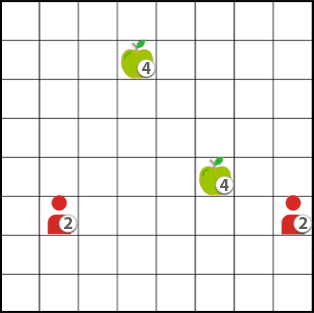}
        \caption{Level-Based Foraging}\label{fig:envs:foraging}
    \endsubfigure\hfill
    \subfigure[t]{0.24\textwidth}
        \includegraphics[width=\linewidth]{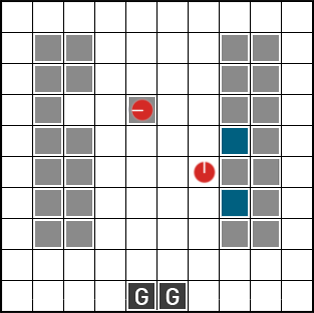}
        \caption{RWARE ($10 \times 11$)}\label{fig:envs:rware}
    \endsubfigure\hfill
    \caption{Environments used in our evaluation. Controlled agents are coloured red.}\label{fig:envs}
\end{figure}

\subsection{Baselines}\label{sec:baselines}


\textbf{Independent Actor-Critic (IAC):}  We compare SEAC to independent learning~\cite{Tan1993Multi-AgentAgents}, in which each agent has its own policy network and is trained separately only using its own experience. IAC uses an actor-critic algorithm for each agent, directly optimising \cref{eq:valueloss,eq:policyloss}; and treating other agents as part of the environment. Arguably, independent learning is one of the most straightforward approaches to MARL and serves as reasonable baseline due to its simplicity.


\textbf{Shared Network Actor-Critic (SNAC):} We also compare SEAC to training a single shared policy among all agents. During execution of the environment, each agents gets a copy of the policy and individually follows it. During training, the sum of policy and value loss gradients are used to optimise the shared parameters. Importance sampling is not required since all trajectories are on-policy. 
Improved performance of our SEAC method would raise the question whether agents simply benefit from processing more data during training. Comparing against this baseline can also show that agents trained using experience sharing are not learning identical policies but instead learn different ones despite being trained on the same collective experience.


\subsection{Algorithm Details}\label{sec:experiments:algorithm}

For all tested algorithms, we implement AC using n-step returns and synchronous environments~\cite{Mnih2016AsynchronousLearning}.
Specifically, 5-step returns were used and four environments were sampled and passed in batches to the optimiser. An entropy regularisation term was added to the final policy loss~\cite{Mnih2016AsynchronousLearning}, computing the entropy of the policy of each individual agent. Hence, the entropy term of agent $i$, given by $H(\pi(o_t^i; \phi_i))$, only considers its own policy.
High computational requirements in terms of environment steps only allowed hyperparameter tuning for IAC on RWARE; all tested AC algorithms use the same hyperparameters (see \cref{sec:additional_details}).
All results presented are averaged across five seeds, with the standard deviation plotted as a shaded area.

\subsection{Results}

\begin{figure}
    \centering
    \subfigure[t]{0.5\linewidth}
        \includegraphics[width=\linewidth]{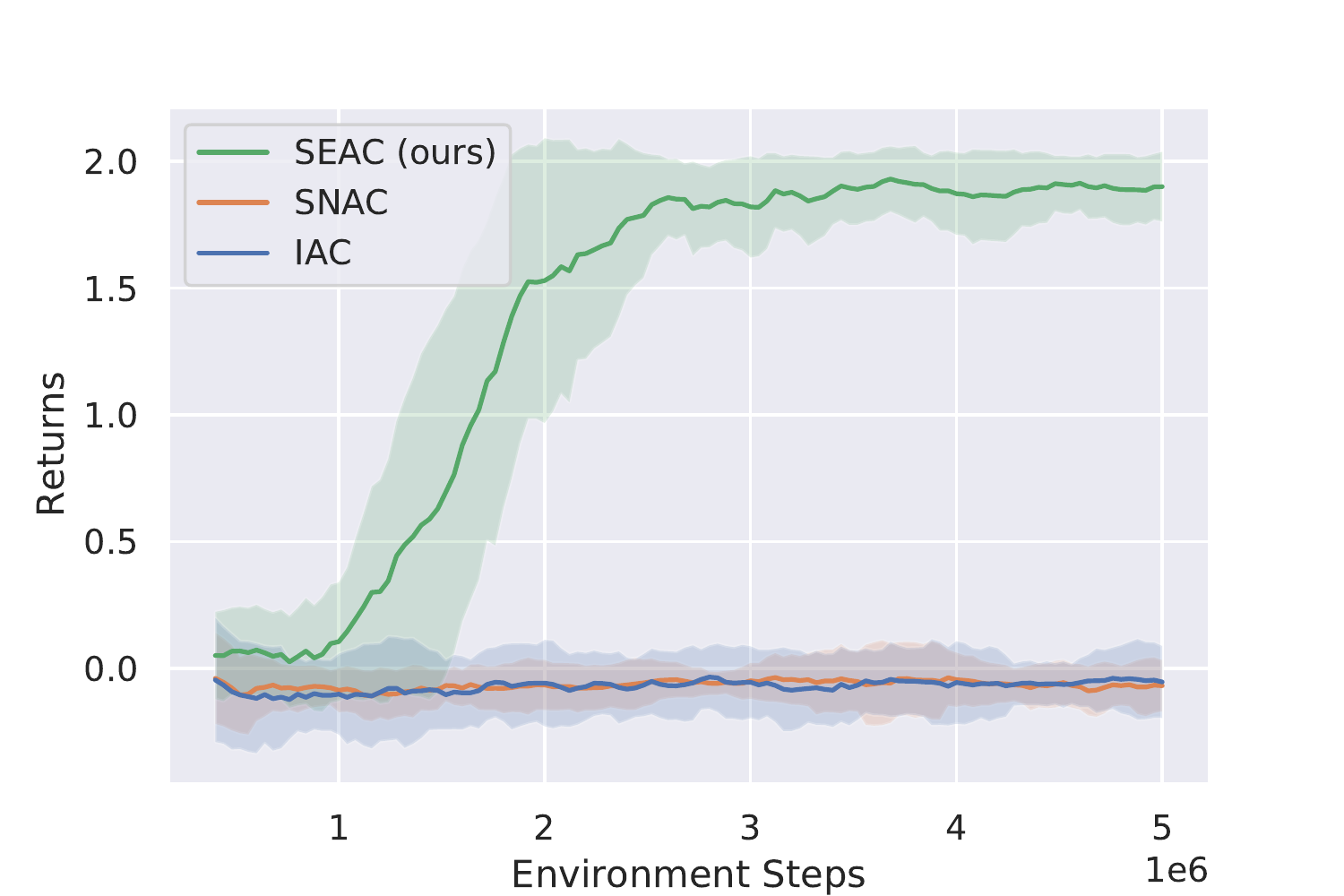}
        \caption{PP, sparse rewards}\label{fig:results:prey}
    \endsubfigure\hfill
    \subfigure[t]{0.5\linewidth}
        \includegraphics[width=\linewidth]{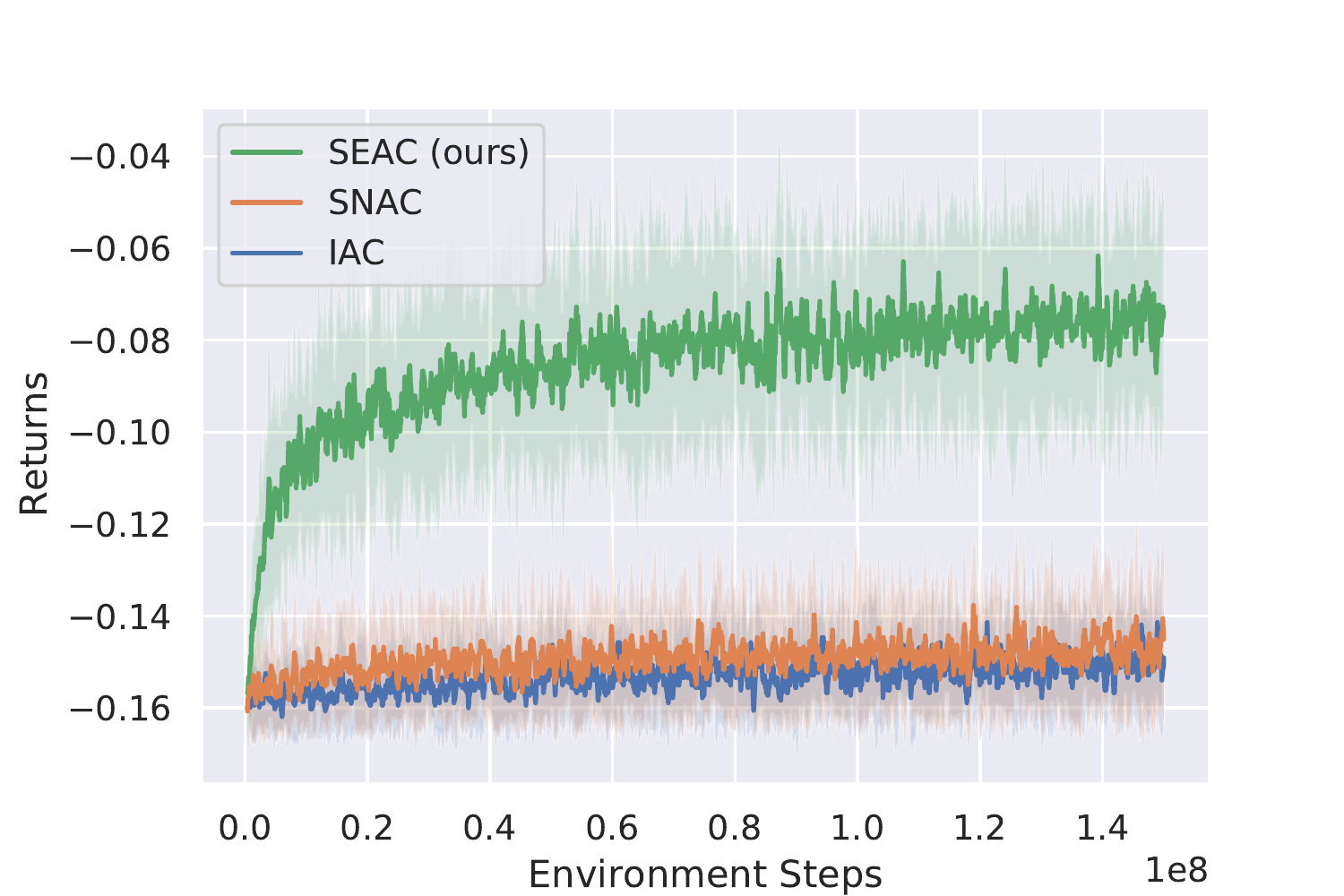}   
        \caption{SMAC with three marines, sparse rewards}\label{fig:results:smac-sparse}
    \endsubfigure\hfill
    \caption{Mean training returns across seeds for sparse-reward variations of PP and SMAC-3m.}\label{fig:results:smac}
\end{figure}

\Cref{fig:results:smac,fig:results:lbf,fig:results:rware,} show the training curves of SEAC, SNAC and IAC for all tested environments. For RWARE and LBF, tasks are sorted from easiest to hardest. 

In the sparse PP task (\Cref{fig:results:prey}) only SEAC learns successfully with consistent learning across seeds while IAC and SNAC are unable to learn to catch the prey.

In SMAC with sparse rewards (\Cref{fig:results:smac-sparse}), SEAC outperforms both baselines. However, with mean returns close to zero, the agents have not learned to win the battles but rather to run away from the enemy. This is not surprising since our experiments (\Cref{tab:sota}) show that even state-of-the-art methods designed for these environments (e.g. QMIX) do not successfully solve this sparsely rewarded task.

For LBF (\Cref{fig:results:lbf}), no significant difference can be observed between SEAC and the two baseline methods IAC and SNAC for the easiest variant (\Cref{fig:results:foraging1}) which does not emphasise exploration. However, as the rewards become sparser the improvement becomes apparent. For increased number of agents, foods and gridsize (\Crefrange{fig:results:foraging2}{fig:results:foraging-coop}), IAC and SNAC converge to significantly lower average returns than SEAC. In the largest grid (\Cref{fig:results:foraging3}), IAC does not show any signs of learning due to the sparsity of the rewards whereas SEAC learns to collect some of the food. We also observe that SEAC tends to converge to its final policy in fewer timesteps than IAC.

\begin{figure}
    \centering
    \subfigure[t]{0.5\linewidth}
        \includegraphics[width=\linewidth]{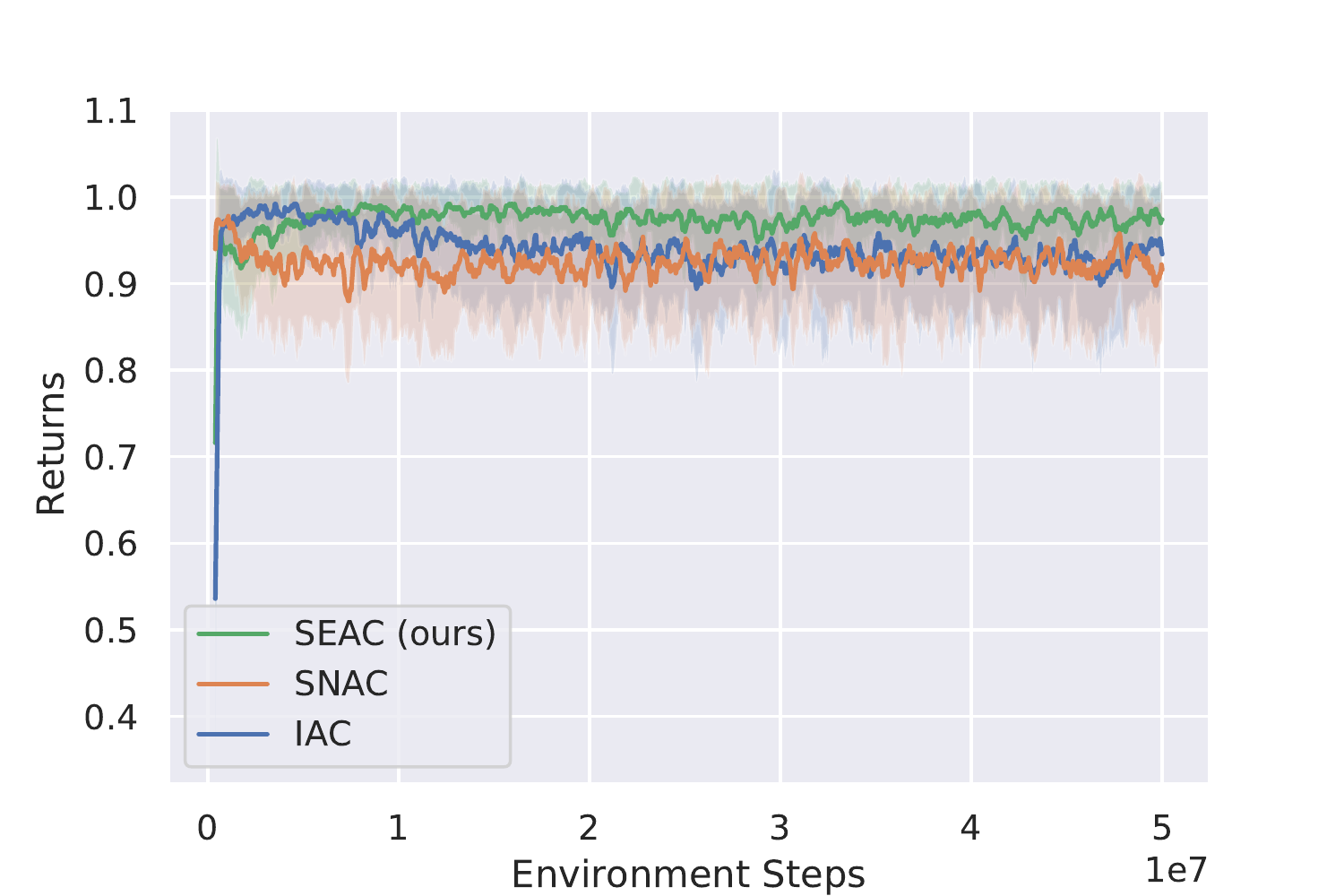}
        \caption{LBF: ($12\times12$), two agents, one food}\label{fig:results:foraging1}
    \endsubfigure\hfill
    \subfigure[t]{0.5\linewidth}
        \includegraphics[width=\linewidth]{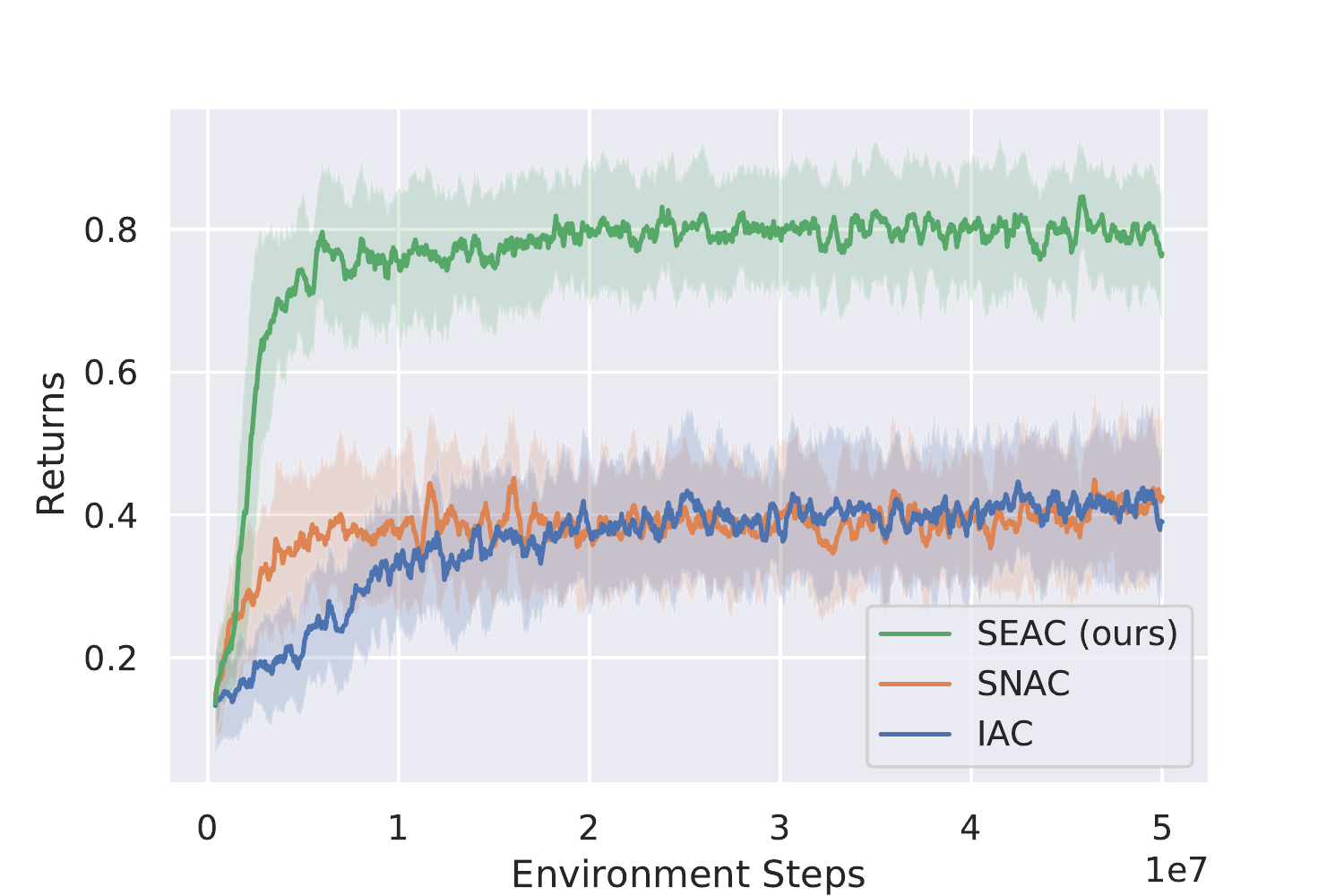}
        \caption{LBF: ($10\times10$), three agents, three foods}\label{fig:results:foraging2}
    \endsubfigure\hfill
    \subfigure[t]{0.5\linewidth}
        \includegraphics[width=\linewidth]{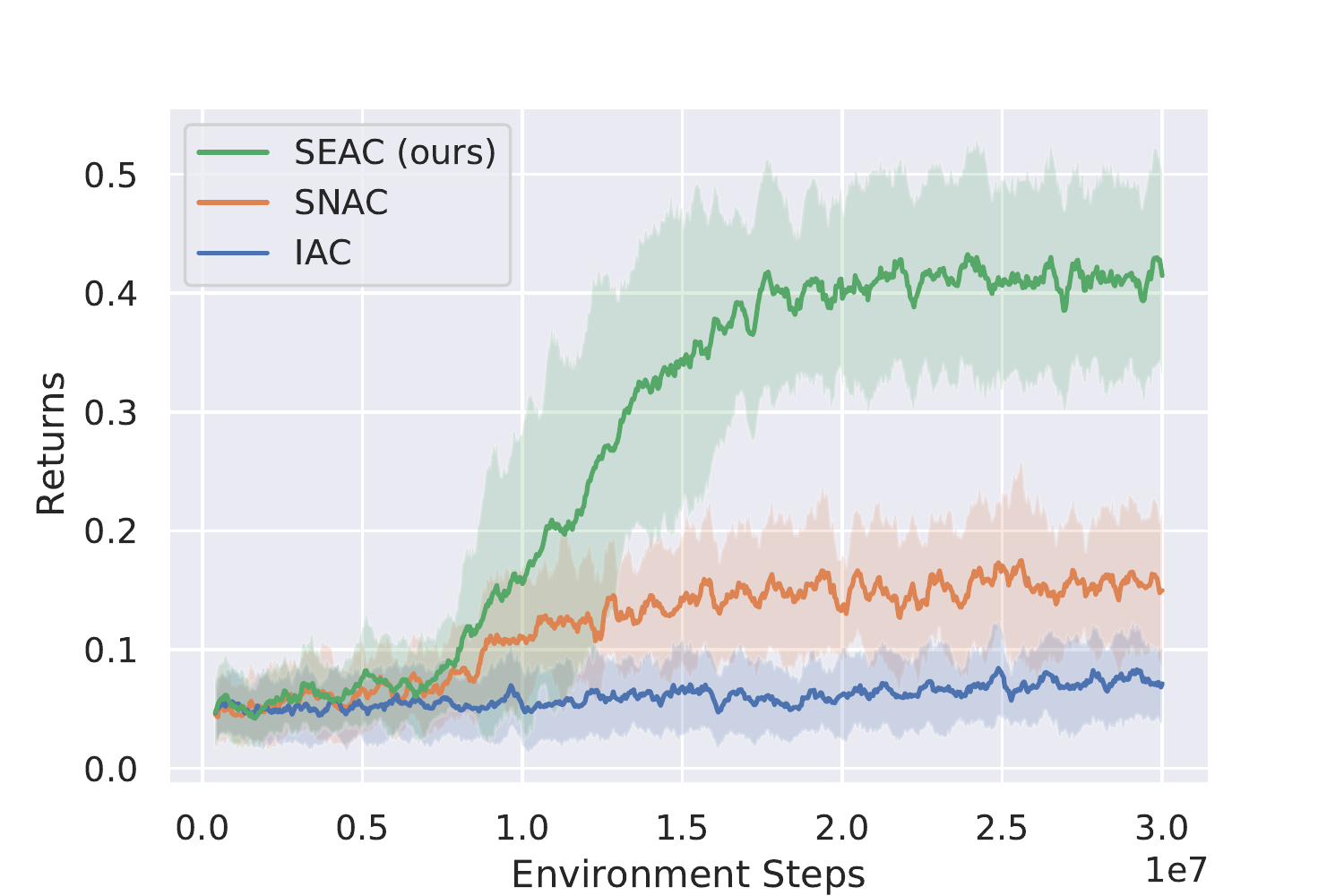}
        \caption{LBF: ($15\times15$), three agents, four food}\label{fig:results:foraging3}
    \endsubfigure
    \subfigure[t]{0.5\linewidth}
        \includegraphics[width=\linewidth]{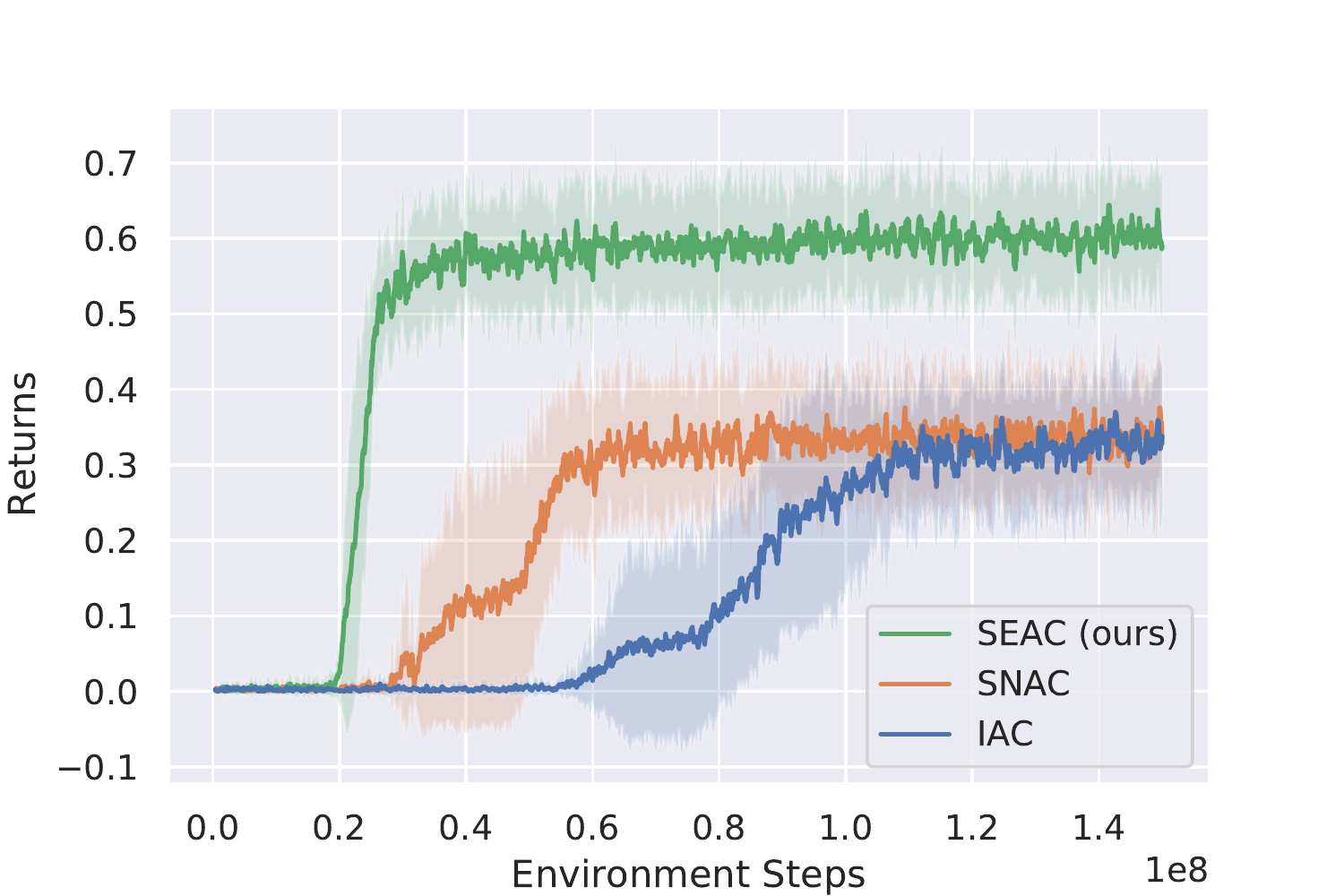}
        \caption{LBF: ($8\times8$), two agents, two foods, cooperative}\label{fig:results:foraging-coop}
    \endsubfigure
    \caption{Mean training returns across seeds on LBF. Tasks are sorted from easiest to hardest.}
    \label{fig:results:lbf}
\end{figure}

\begin{figure}
    \centering
    \subfigure[t]{0.5\linewidth}
        \includegraphics[width=\linewidth]{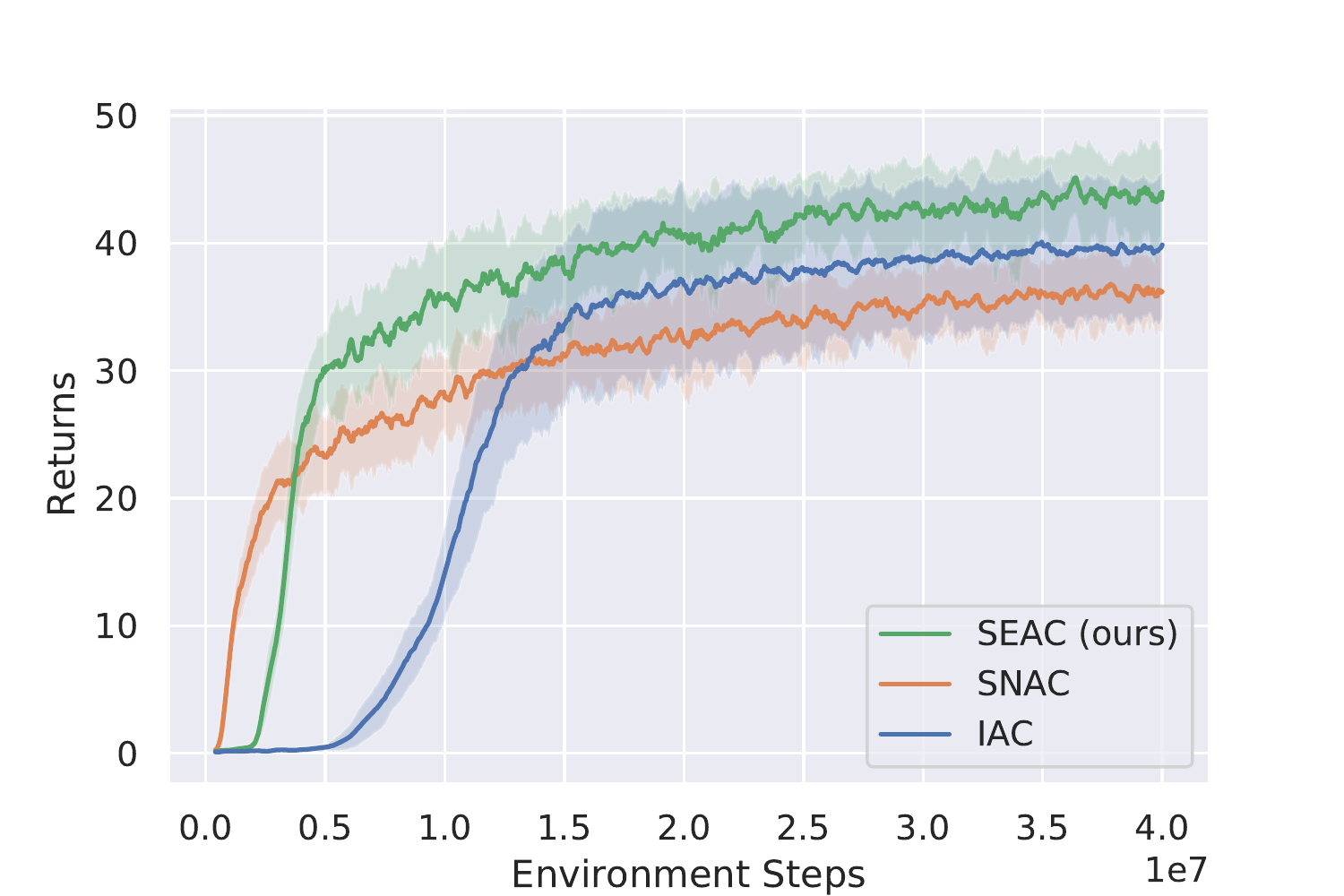}
        \caption{RWARE: ($10\times11$), four agents}\label{fig:results:tiny4}
    \endsubfigure\hfill
    \subfigure[t]{0.5\linewidth}
        \includegraphics[width=\linewidth]{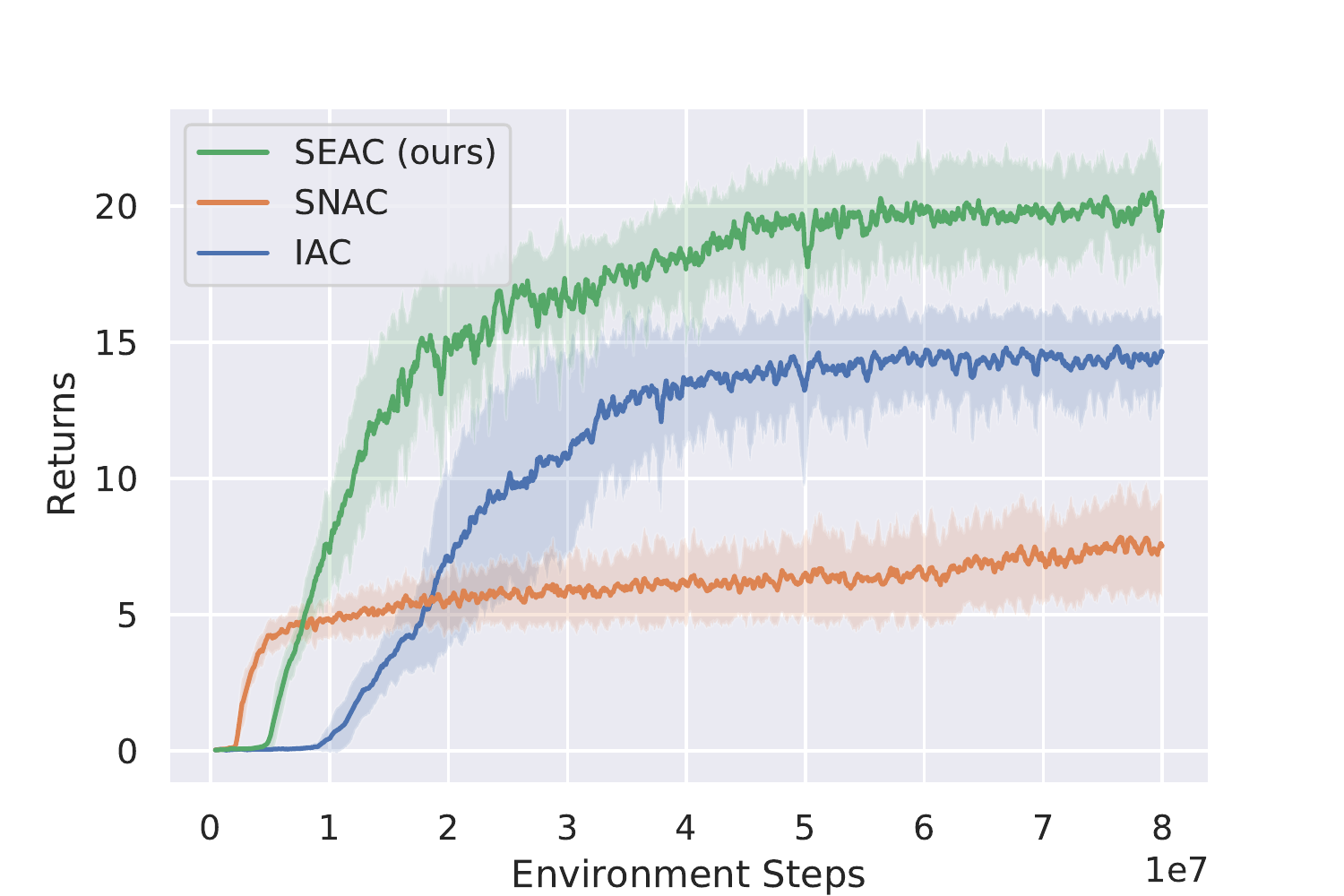}
        \caption{RWARE: ($10\times11$), two agents}\label{fig:results:tiny2}
    \endsubfigure\hfill
    \subfigure[t]{0.5\linewidth}
        \includegraphics[width=\linewidth]{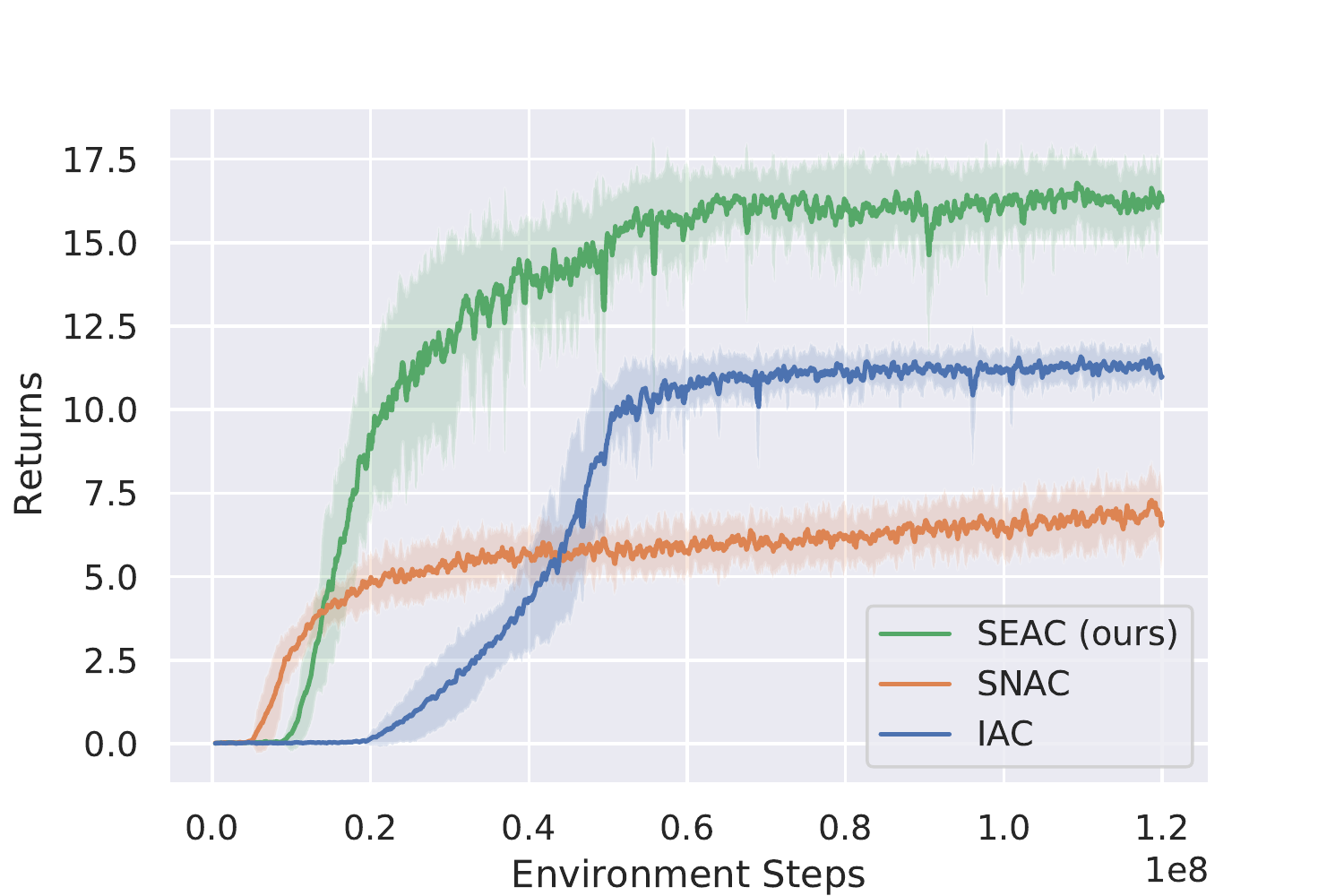}
        \caption{RWARE: ($10\times11$), two agents, hard}\label{fig:results:tiny2hard}
    \endsubfigure
    \subfigure[t]{0.5\linewidth}
        \includegraphics[width=\linewidth]{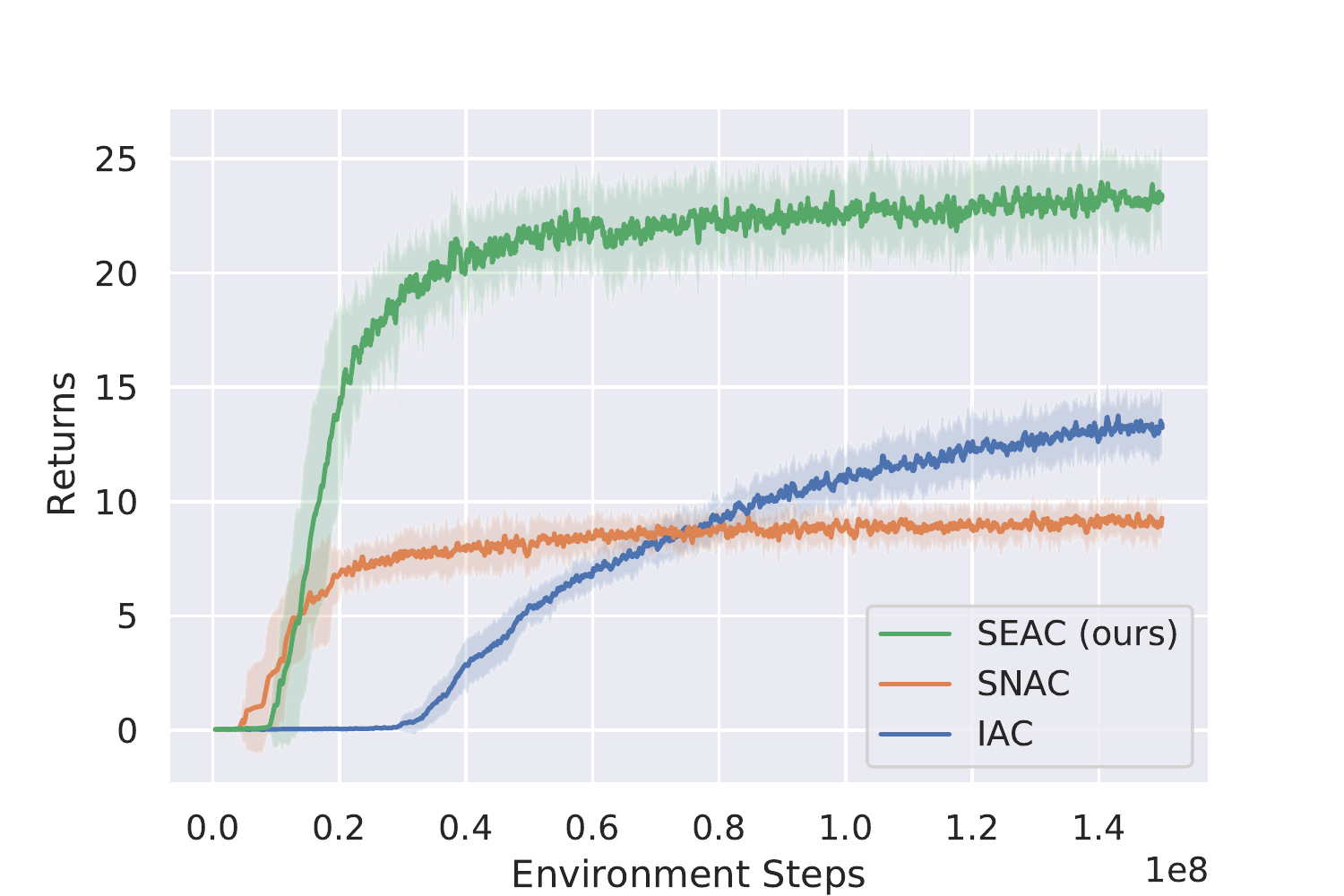}
        \caption{RWARE: ($10\times20$), four agents}\label{fig:results:small}
    \endsubfigure
    \caption{Mean training returns across seeds on RWARE. Tasks are sorted from easiest to hardest.}
    \label{fig:results:rware}
\end{figure}

In RWARE (\Cref{fig:results:rware}), results are similar to LBF. Again, the two baseline methods IAC and SNAC converge to lower average returns than SEAC for harder tasks. In the hardest task (\Cref{fig:results:small}), SEAC converges to final mean returns $\approx70\%$ and $\approx160\%$ higher than IAC and SNAC, respectively, and again converges in fewer steps than IAC.

In \Cref{tab:sota} we also present the final evaluation returns of three state-of-the-art MARL algorithms (QMIX~\cite{Rashid2018QMIX:Learning}, MADDPG~\cite{LoweMulti-AgentEnvironments}, and ROMA~\cite{wang2020roma}) on a selection of tasks. These algorithms show no signs of learning in most of these tasks. The only exceptions are MADDPG matching the returns of SEAC in the sparse PP task and QMIX performing comparably to SEAC in the cooperative LBF task. QMIX and ROMA assume tasks to be fully-cooperative, i.e.\ all agents receive the same reward signal. Hence, in order to apply the two algorithms, we modified non-cooperative environments to return the sum of all individual agent returns as the shared reward. While shared rewards could make learning harder, we also tested IAC in the easiest variant of RWARE and found that it learned successfully even with this reward setting.

We also evaluate \emph{Shared Experience Q-Learning}, as described in \Cref{sec:appendix_seql}, and Independent Q-Learning~\cite{Tan1993Multi-AgentAgents} based on DQN~\cite{mnih2015human}. In some sparse reward tasks, shared experience did reduce variance, and improved total returns. However, less impact has been observed through the addition of sharing experience to this off-policy algorithm compared to SEAC. Results can be found in \Cref{sec:appendix_seql}.

In terms of computational time, sharing experience with SEAC increased running time by less than $3\%$ across all environments compared to IAC. More details can be found in \Cref{sec:additional_details}.

\begin{table}
    \caption{Final mean evaluation returns across five random seeds with standard deviation on a selection of tasks. Highest means per task (within one standard deviation) are shown in bold.}\label{tab:sota}
    \resizebox{\textwidth}{!}{
    \begin{tabular}{@{}lrrrrrr@{}}
    \toprule
                             & IAC                   & SNAC          & SEAC (ours)          & QMIX                  & MADDPG                & ROMA\\ \midrule
    PP (sparse)              & -0.04 ±0.13           & -0.04 ±0.1    & \textbf{1.93 ±0.13}  & 0.05 ±0.07            & \textbf{2.04 ±0.08}   & 0.04 ±0.07  \\
    SMAC-3m (sparse)         & -0.13 ±0.01           & -0.14 ±0.02   & \textbf{-0.03 ±0.03} & \textbf{0.00 ±0.00}   &\textbf{-0.01 ±0.01}   & \textbf{0.00 ±0.00} \\
    LBF-(15x15)-3ag-4f       & 0.13 ±0.04            & 0.18 ±0.08    & \textbf{0.43 ±0.09}  & 0.03 ±0.01            & 0.01 ±0.02            & 0.03 ±0.02  \\
    LBF-(8x8)-2ag-2f-coop    & 0.37 ±0.10            & 0.38 ±0.10    & \textbf{0.64 ±0.08}  & \textbf{0.79 ±0.31}   & 0.01 ±0.02            & 0.01 ±0.02  \\
    RWARE-(10x20)-4ag        & 13.75 ±1.26           & 9.53 ±0.83    & \textbf{23.96 ±1.92} & 0.00 ±0.00            & 0.00 ±0.00            & 0.00 ±0.00  \\
    RWARE-(10x11)-4ag        & \textbf{40.10 ±5.60}  & 36.79 ±2.36   & \textbf{45.11 ±2.90} & 0.00 ±0.00            & 0.00 ±0.00            & 0.01 ±0.01  \\
    \bottomrule
    \end{tabular}
    }
\end{table}

\subsection{Analysis}

\begin{figure}[t]
    \centering
    \begin{minipage}{.5\textwidth}
        \centering
        \captionsetup{width=.9\linewidth}
        \includegraphics[width=\linewidth]{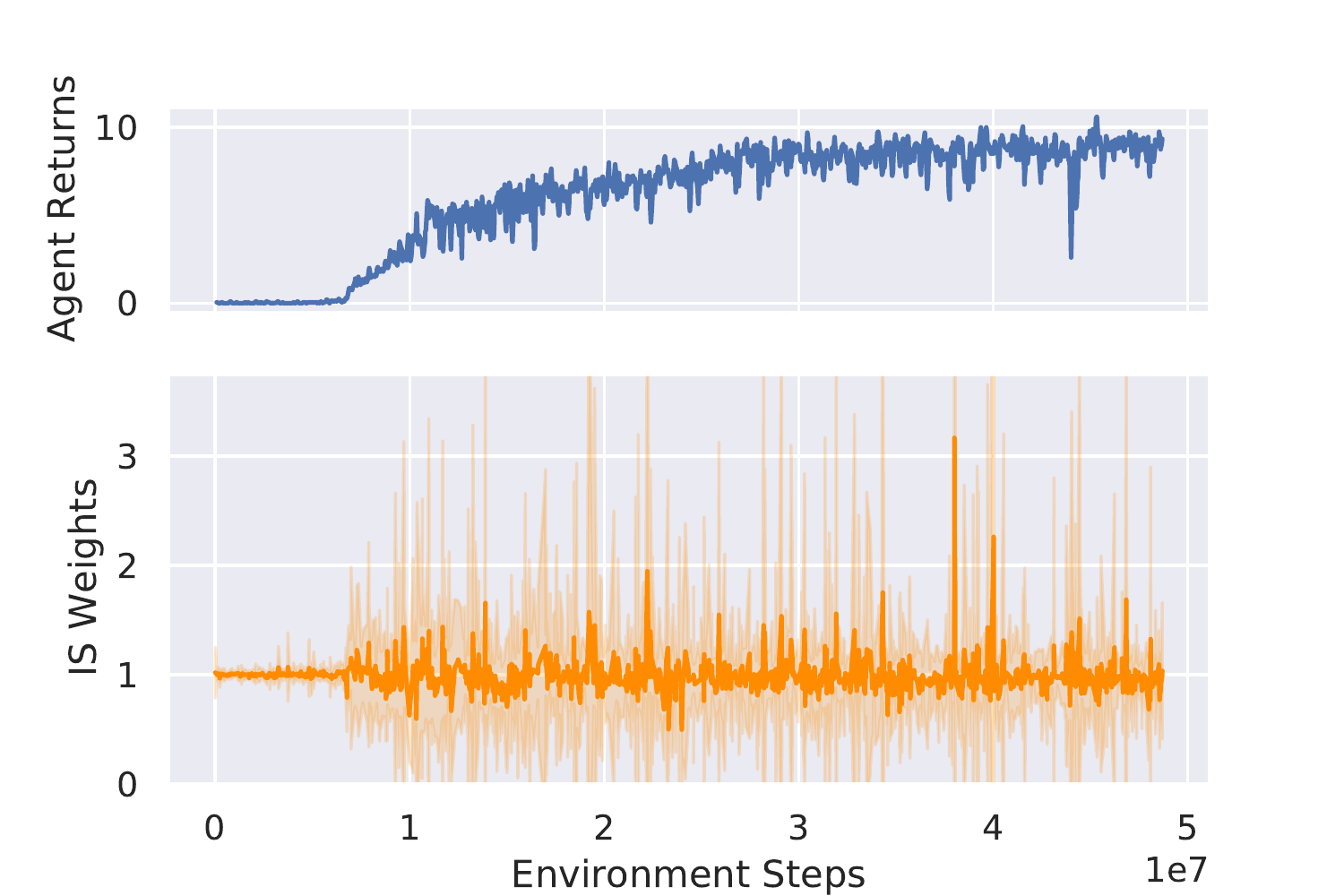}
        \captionof{figure}{Importance weights of one SEAC agent in RWARE, ($10\times11$), two agents, hard}
        \label{fig:importance-weights}
    \end{minipage}%
    \begin{minipage}{.5\textwidth}
        \centering
        \captionsetup{width=.9\linewidth}
        \includegraphics[width=\linewidth]{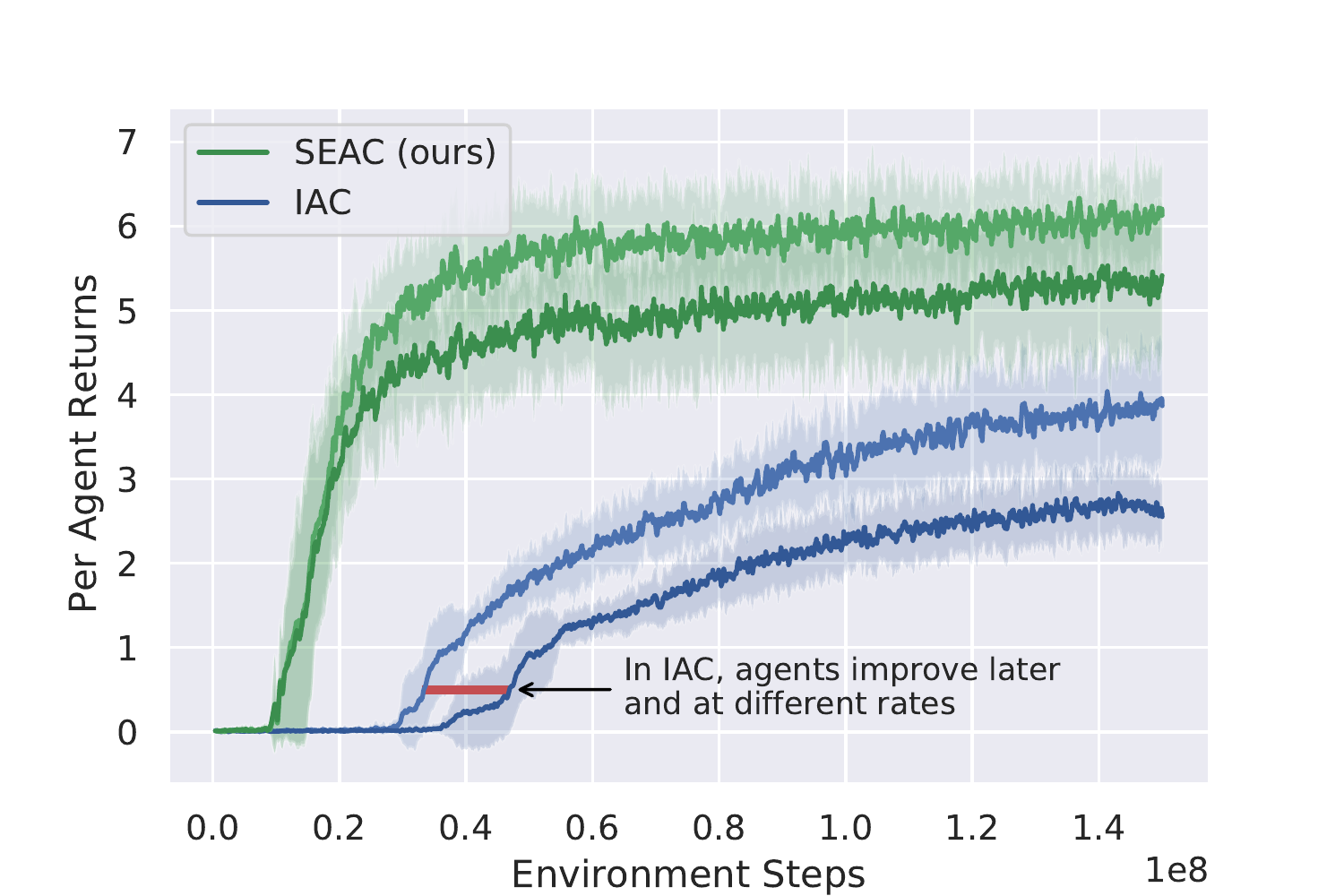}
        \captionof{figure}{Best vs.\ Worst performing agents on RWARE, ($10\times20$), four agents}
        \label{fig:best-vs-worst}
    \end{minipage}%
\end{figure}

Similar patterns can be seen for the different algorithms across all tested environments. It is not surprising that IAC requires considerably more environment samples to converge, given that the algorithm is less efficient in using them; IAC agents only train on their own experience. This is further evident when noticing that in RWARE (\crefrange{fig:results:tiny4}{fig:results:small}) the learning curve of SEAC starts moving upwards in roughly $\sfrac{1}{N}$ the timesteps compared to IAC, where $N$ refers to the number of agents. Also, it is not surprising that SNAC does not achieve as high returns after convergence: sharing a single policy across all agents impedes their ability to coordinate or develop distinct behaviours that lead to higher returns.

We conjecture that SEAC converges to higher final returns due to agents improving at similar rates when sharing experiences, combined with the flexibility to develop differences in policies to improve coordination. We observe that SEAC is able to learn similarly quickly to SNAC because the combined local gradients provide a very strong learning direction. However, while SNAC levels off at some point due to the use of identical policies, which limit the agents' ability to coordinate, SEAC can continue to explore and improve because agents are able to develop different policies to further improve coordination. \Cref{fig:importance-weights} shows that encountered importance weights during SEAC optimisation are centred around 1, with most weights staying in the range $[0.5; 1.5]$. This indicates that the agents indeed learn similar but not identical policies. The divergence of the policies is attributed to the random initialisation of networks, along with the agent-centred entropy factor (\Cref{sec:experiments:algorithm}). The range of the importance weights also shows that, in our case, importance sampling does not introduce significant instability in the training. The latter is essential for learning since importance weighting for off-policy RL is known to suffer from significant instability and high variance through diverging policies~\cite{Sutton1998ReinforcementIntroduction,precup2000eligibility}.

In contrast, we observe that IAC starts to improve at a much later stage than SEAC because agents need to explore for longer, and when they start improving it is often the case that one agent improves first while the other agents catch up later, which can severely impede learning. \Cref{fig:best-vs-worst} shows that agents using IAC end up learning at different rates, and the slowest one ends up with the lowest final returns. In learning tasks that require coordination, an agent being ahead of others in its training can impede overall training performance.

We find examples of agents learning at different rates being harmful to overall training in all our tested environments. In RWARE, an agent that learns to fulfil requests can make the learning more difficult for others by delivering all requests on its own. Agents with slightly less successful exploration have a harder time learning a rewarding policy when the task they need to perform is constantly done by others. In LBF, agents can choose to cooperate to gather highly rewarding food or focus on food that can be foraged independently. The latter is happening increasingly often when an agent is ahead in the learning curve as others are still aimlessly wandering in the environment. In the PP task, the predators must approach the prey simultaneously, but this cannot be the case when one predator does not know how to do so. In the SMAC-3m task, a single agent cannot be successful if its team members do not contribute to the fight. The agent would incorrectly learn that fighting is not viable and therefore prefer to run from the enemy, which however is not an optimal strategy.


\section{Conclusion}

This paper introduced SEAC, a novel multi-agent actor-critic algorithm in which agents learn from the experience of others. In our experiments, SEAC outperformed independent learning, shared policy training, and state-of-the-art MARL algorithms in ten sparse-reward learning tasks, across four environments, demonstrating
improved sample efficiency and final returns. We discussed a theme commonly found in MARL environments: agents learning at different rates impedes exploration, leading to sub-optimal policies. 
SEAC overcomes this issue by combining the local gradients and concurrently learning similar policies for all agents, but it also benefits from not being restricted to identical policies, allowing for better coordination and exploration.

Sharing experience is appealing especially due to its simplicity. We showed that barely any additional computational power, nor any extra parameter tuning are required and no additional networks are introduced. Therefore, its use should be considered in all environments that fit the requirements.

Future work could aim to relax the assumptions made for tasks SEAC can be applied to 
and evaluate in further multi-agent environments. Also, our work focused on the application of experience sharing to independent actor-critic. Further analysis of sharing experience as a generally applicable concept for MARL and its impact on a variety of MARL algorithms is left for future work.


\section*{Broader Impact}

Multi-agent deep reinforcement learning has potential applications in areas such as autonomous vehicles, robotic warehouses, internet of things, smart grids, and more. Our research could be used to improve reinforcement learning models in such applications. However, it must be noted that real-world application of MARL algorithms is currently not viable due to open problems in AI explainability, robustness to failure cases, legal and ethical aspects, and other issues, which are outside the scope of our work. 

That being said, improvements in MARL could lead to undue trust in RL models; having models that work well does not translate to models that are safe or which can be broadly used. Agents trained with these methods need to be thoroughly studied before being used in production. However, if these technologies are indeed used responsibly, they can improve several aspects of modern society such as making transportation safer, or performing jobs that might pose risks to humans.

\section*{Funding Disclosure}
This research was in part financially supported by the UK EPSRC Centre for Doctoral Training in Robotics and Autonomous Systems (F.C.), the Edinburgh University Principal’s Career Development Scholarship (L.S.), and personal grants from the Royal Society and the Alan Turing Institute (S.A.).

\AtNextBibliography{\small} 
\printbibliography

\clearpage
\appendix

\section{Environments}\label{apdx:envs}
\subsection{Multi-Robot Warehouse}\label{sec:warehouse}
The multi-robot warehouse environment (\Cref{fig:rware}) simulates a warehouse with robots moving and delivering requested goods. In real-world applications~\cite{Wurman2008}, robots pick-up shelves and deliver them to a workstation. Humans assess the content of a shelf, and then robots can return them to empty shelf locations.
In this simulation of the environment, agents control robots and the action space for each agent is
\begin{equation*}
    A=\{\text{Turn Left, Turn Right, Forward, Load/Unload Shelf}\}
\end{equation*}
Agents can move beneath shelves when they do not carry anything, but when carrying a shelf, agents must use the corridors visible in \Cref{fig:rware}.

\begin{figure}[t]
    \subfigure[t]{0.25\textwidth}
        \includegraphics[width=\linewidth]{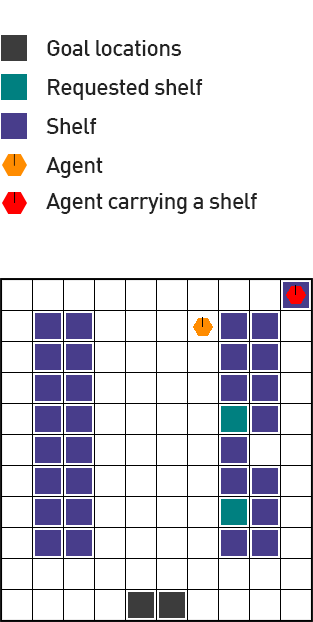}
        \caption{Tiny size, two agents}\label{fig:rware:tiny}
    \endsubfigure\hfill
    \subfigure[t]{0.25\textwidth}
        \includegraphics[width=\linewidth]{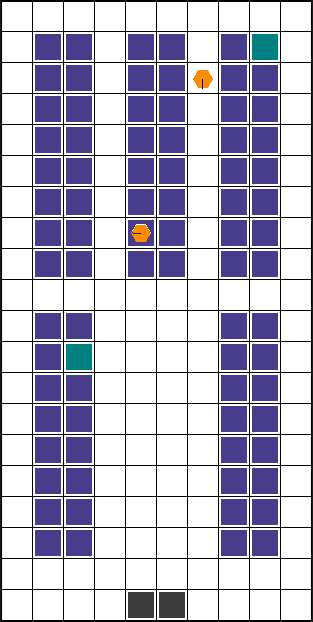}
        \caption{Small size, two agents}\label{fig:rware:small}
    \endsubfigure\hfill
    \subfigure[t]{0.40\textwidth}%
        \includegraphics[width=\linewidth]{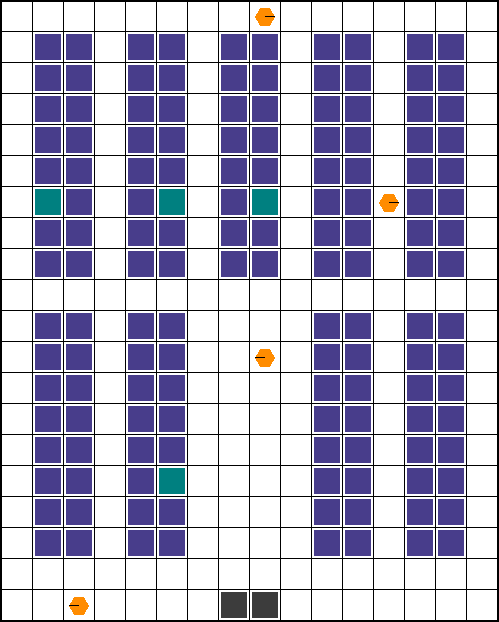}
        \caption{Medium size, four agents}\label{fig:rware:medium}
    \endsubfigure
    \caption{Three size variations of the multi-robot warehouse environment.}\label{fig:rware}
\end{figure}

The observation of an agent consists of a $3\times3$ square centred on the agent. It contains information about the surrounding agents (location/rotation) and shelves. 

At each time a fixed number of shelves $R$ is requested. When a requested shelf is brought to a goal location (dark squares in \cref{fig:rware}), another currently not requested shelf is uniformly sampled and added to the current requests. Agents are rewarded for successfully delivering a requested shelf to a goal location, with a reward of $1$. A major challenge in this environments is for agents to deliver requested shelves but also afterwards finding an empty shelf location to return the previously delivered shelf. Agents need to put down their previously delivered shelf to be able to pick up a new shelf. This leads to a very sparse reward signal.

Since this is a collaborative task, as a performance metric we use the sum of the undiscounted returns of all the agents.

The multi-robot warehouse task is parameterised by:
\begin{itemize}
    \item The size of the warehouse which is preset to either tiny ($10 \times 11$), small ($10 \times 20$), medium ($16 \times 20$), or large ($16 \times 29$). 
    \item The number of agents $N$. 
    \item The number of requested shelves $R$. By default $R=N$, but easy and hard variations of the environment use $R=2N$ and $R=\sfrac{N}{2}$, respectively.
\end{itemize}

Note that $R$ directly affects the difficulty of the environment. A small $R$, especially on a larger grid, dramatically affects the sparsity of the reward and thus exploration: randomly bringing the correct shelf becomes increasingly improbable.

\subsection{Level-Based Foraging}\label{sec:foraging}
The level-based foraging environment (\Cref{fig:lbforaging}) represents a mixed cooperative-competitive game~\cite{Albrecht2013ASystems}, which focuses on the coordination of the agents involved. Agents navigate a grid world and collect food by cooperating with other agents if needed.

\begin{figure}
    \subfigure[t]{0.24\textwidth}
        \includegraphics[width=\linewidth]{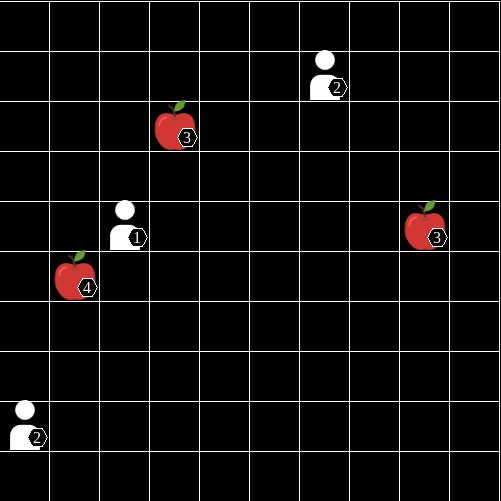}
        \caption{Foraging-10x10-3p-3f}
    \endsubfigure\hfill
    \subfigure[t]{0.24\textwidth}
        \includegraphics[width=\linewidth]{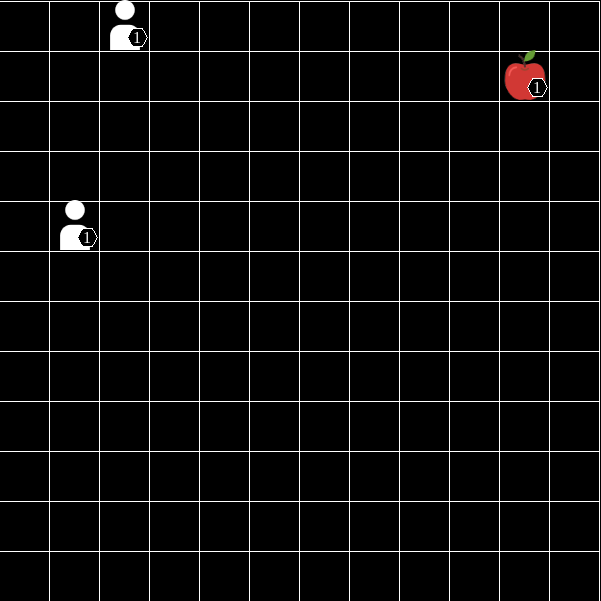}
        \caption{Foraging-12x12-2p-1f}
    \endsubfigure\hfill
    \subfigure[t]{0.24\textwidth}
        \includegraphics[width=\linewidth]{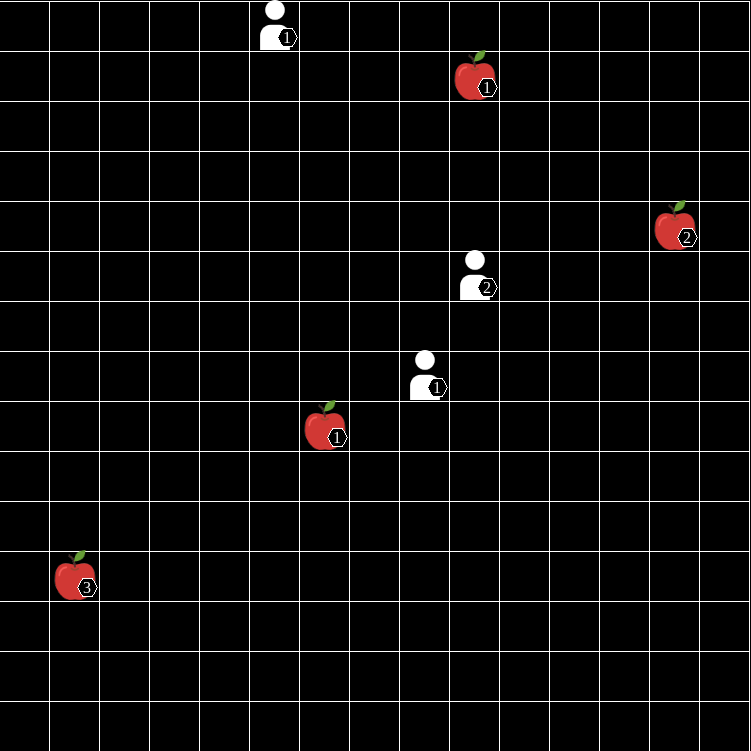}
        \caption{Foraging-15x15-3p-4f}
    \endsubfigure\hfill
    \subfigure[t]{0.24\textwidth}%
        \includegraphics[width=\linewidth]{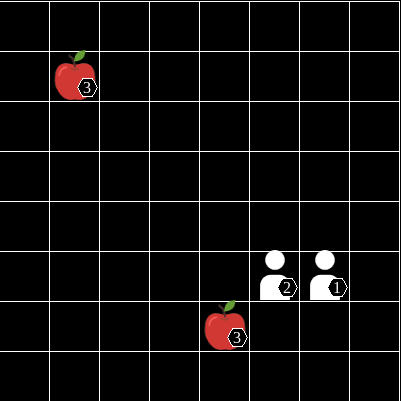}
        \caption{Foraging-8x8-2p-2f-coop}\label{fig:lbforaging:coop}
    \endsubfigure
    \caption{Four variations of level based foraging used in this work.}\label{fig:lbforaging}
\end{figure}

More specifically, agents and food are randomly scattered in the grid world, and each is assigned a level. 
Agents can navigate in the environment and attempt to collect food placed next to them. The collection of food is successful only if the sum of the levels of all agents involved in collecting at the same time is equal to or higher than the level of the food. Agents are rewarded proportional to the level of food they took part in collecting. Episodes are terminated once all food has been collected or the maximum episode length of $25$ timesteps is reached.

We are using full observability for this environment, meaning agents observe the locations and levels of all entities in the map. Each agent can attempt to move in all four directions and attempt to load adjacent food, for a total of five actions. After successfully loading a food, agents are rewarded:
\begin{equation*}
    r^i = \frac{FoodLevel * AgentLevel}{\sum{FoodLevels}\sum{LoadingAgentsLevel}}
\end{equation*}
This normalisation ensures that the sum of the agent returns on a solved episode equals to one.

Note that the final variant, \Cref{fig:lbforaging:coop}, is a fully-cooperative environment. Food levels are always equal to the sum of all agents' levels, requiring all agents to load simultaneously, and thus sharing the reward. 

\section{Additional Experimental Details}\label{sec:additional_details}
\begin{wraptable}{r}{0.5\textwidth}
    \caption{Hyperparameters used for implementation of SEAC, IAC and SNAC}\label{tab:hyperparams}
    \centering
    \begin{tabular}{@{}lr@{}}
    \toprule
    Hyperparameter     & Value\\ \midrule
    learning rate      & $3e^{-4}$                 \\
    network size       & $64\times64$              \\
    adam epsilon       & 0.001                     \\
    gamma              & 0.99                      \\
    entropy coef       & 0.01                      \\
    value loss coef    & 0.5                       \\
    GAE                & False \\
    grad clip          & 0.5                       \\
    parallel processes & 4                         \\
    n-steps            & 5                         \\
    $\lambda$ (\Cref{eq:policyloss_shared,eq:valueloss_shared})   & 1.0                         \\ \bottomrule
    \end{tabular}
\end{wraptable}

Our implementations of IAC, SEAC, and SNAC closely follow A2C~\cite{Mnih2016AsynchronousLearning}, using n-step returns and parallel sampled environments. \Cref{tab:hyperparams} contains the hyperparameters used in the experiments. Hyperparameters for MADDPG, QMIX and ROMA were optimised using a grid search over learning rate, exploration rate and batch sizes with the grid centred on the hyperparameters used in the original papers and parameter performance tested in all used environments.

\Cref{tab:comp_time} contains process time required for running IAC and SEAC. Timings were measured on a $6^{th}$ Gen Intel i7 @ 4.6 Ghz running Python 3.7 and PyTorch 1.4. The average time for running and training on 100,000 environment iterations is displayed. Only process time (the time the program was active in the CPU) was measured, rounded to seconds. Often, the bottleneck is the environment and not the network update and as such, more complex and slower simulators, such as SMAC, show a lower percentage difference between algorithms.

\begin{table}[t]\centering
    \caption{Measured mean process time (mins:secs) required for \num{100000} timesteps.  }\label{tab:comp_time}
    \begin{tabular}{@{}lrrr@{}}
    \toprule
                               & \multicolumn{1}{l}{IAC} & \multicolumn{1}{l}{SEAC} & \multicolumn{1}{l}{\% increase} \\ \midrule
    Foraging-10x10-3p-3f-v0    & 2:00                    & 2:04                     & 3.86\%                        \\
    Foraging-12x12-2p-1f-v0    & 1:22                    & 1:24                     & 2.94\%                        \\
    Foraging-15x15-3p-4f-v0    & 2:01                    & 2:06                     & 3.90\%                        \\
    Foraging-8x8-2p-2f-coop-v0 & 1:21                    & 1:24                     & 3.78\%                        \\
    rware-tiny-2ag-v1          & 1:41                    & 1:43                     & 1.65\%                        \\
    rware-tiny-2ag-hard-v1     & 2:05                    & 2:09                     & 2.97\%                        \\
    rware-tiny-4ag-v1          & 2:49                    & 2:53                     & 2.25\%                        \\
    rware-small-4ag-v1         & 2:50                    & 2:55                     & 2.44\%                        \\
    Predator Prey              & 2:44                    & 2:49                     & 3.39\%                        \\
    SMAC (3m)                  & 6:23                    & 6:25                     & 0.38\%                        \\ \bottomrule
    \end{tabular}
\end{table}

\begin{figure}[t]
    \centering
    \subfigure[t]{0.5\linewidth}
        \includegraphics[width=\linewidth]{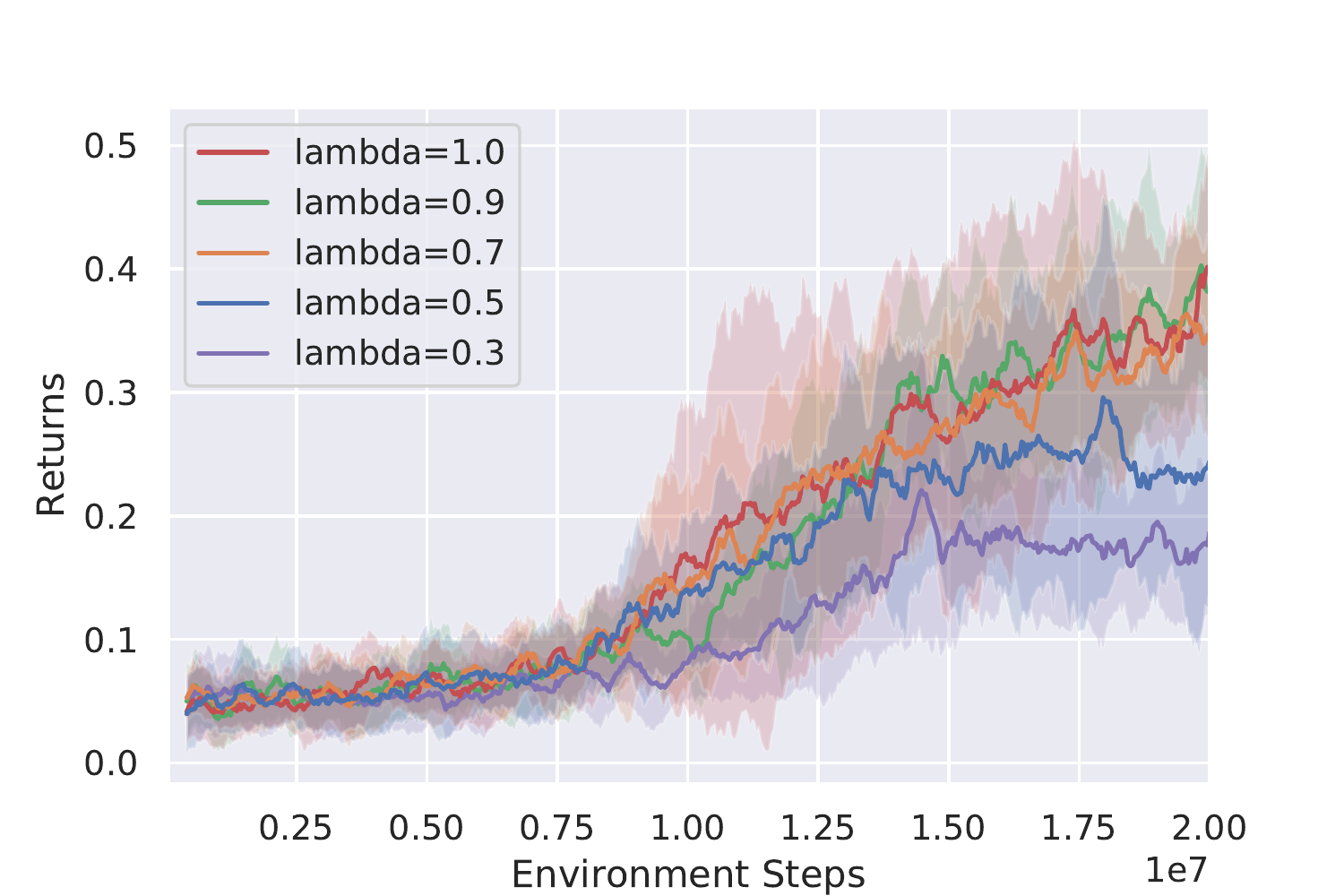}   
        \caption{LBF ($15 \times 15$), 3 agents, 4 foods}
    \endsubfigure\hfill
     \subfigure[t]{0.5\linewidth}
        \includegraphics[width=\linewidth]{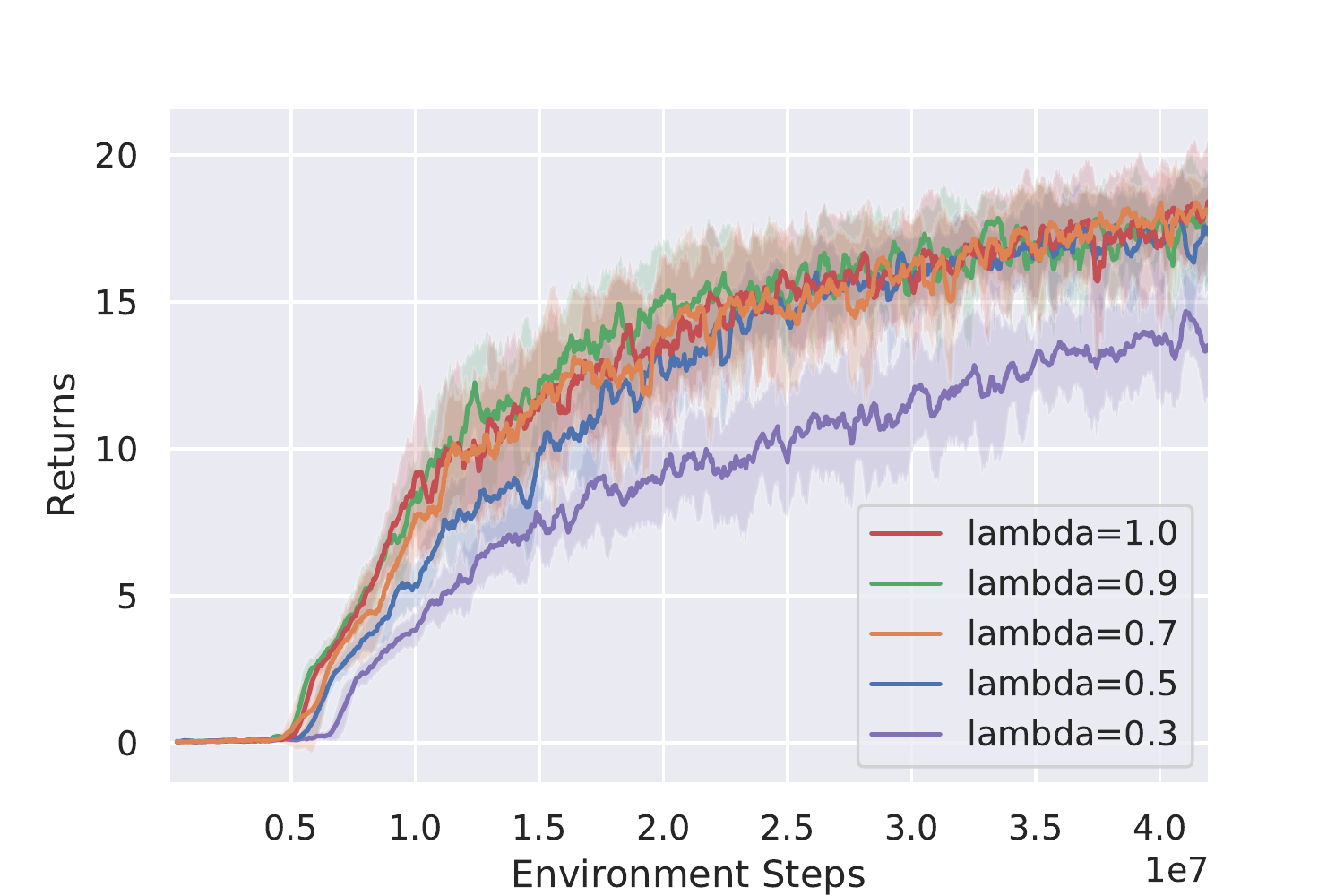}
        \caption{RWARE ($10 \times 20$), two agents}
    \endsubfigure\hfill
    \caption{Training returns with different values of $\lambda$ in SEAC}\label{fig:appendix:lambdas}
\end{figure}

\Cref{fig:appendix:lambdas} shows the training returns with respect to different $\lambda$ values being applied in SEAC. We find that SEAC is not sensitive to tuning of the hyperparameter $\lambda$ with similar performance across a wide range of values. Much lower values for $\lambda$ closer to $0$ lead to decreased performance, eventually converging to IAC for $\lambda = 0$.

For calculation of evaluation returns (\cref{tab:sota}), the best saved models per seed were selected and evaluated for 100 episodes. During evaluation, QMIX and ROMA use $\epsilon=0$, while MADDPG and AC algorithms apply stochastic policies.

\section{SEAC Loss Derivation}\label{appendix:loss_derivation}
We provide the following derivation of SEAC policy loss, as shown in \Cref{eq:policyloss_shared}, for a fully observable two-agent Markov game
$$ \mathcal{M} = \left( \cN = \{1, 2\}, \cS, (A^1, A^2), \cP, (\cR^1, \cR^2)\right) $$
As per \Cref{sec:background}, let $\cA = A^1 \times A^2$ be the joint action space and $A=A^1=A^2$.

In the following, we use $\pi_1$ and $\pi_2$ to denote the policy of agent $1$ and agent $2$ which are conditioned on parameters $\phi_1$ and $\phi_2$, respectively. We use $V^1$ and $V^2$ to denote the state value function of agents $1$ and $2$ which are conditioned on parameters $\theta_1$ and $\theta_2$.



In order to account for different action distributions under policies $\pi_1$ and $\pi_2$, we use importance sampling (IS) defined for any function $g$ over actions
\begin{equation*}
    \E_{a\sim \pi_1(a|s)}\left[ g(a) \right] = \E_{a\sim \pi_2(a|s')}\left[\frac{\pi_1(a|s)}{\pi_2(a|s')} g(a) \right]
\end{equation*}

which can be derived as follows
\begin{equation*}
    \E_{a\sim \pi_1(a|s)}\left[ g(a) \right] = \int_a \pi_1(a|s) g(a) da = \int_a \frac{\pi_2(a|s')}{\pi_2(a|s')} \pi_1(a|s) g(a) da = \E_{a\sim \pi_2(a|s')}\left[  \frac{\pi_1(a|s)}{\pi_2(a|s')} g(a) \right]
\end{equation*}

\begin{assumption}[Reward Independence Assumption: A1]
\label{assumption:reward_independence}
    We assume that an agent perceives the rewards as dependent only on its own action. Other agents are perceived as part of the environment.
    \begin{equation*}
        \begin{split}
            \forall s, s' \in \cS: \forall a \in A: \hat{\cR}^1(s, a, s') &=  \cR^1(s,(a,\cdot), s')\\
            \forall s, s' \in \cS: \forall a \in A: \hat{\cR}^2(s, a, s') &=  \cR^2(s,(\cdot,a), s')
        \end{split}
    \end{equation*}
\end{assumption}

\begin{assumption}[Symmetry Assumption: A2]\label{assumption:symmetry}
    We assume there exists a function $f: \cS \mapsto \cS$ such that
    \begin{equation*}
        \begin{split}
        &\forall s, s' \in \cS: \forall (a_1, a_2) \in \cA: \cR^1(f(s), (a_2, a_1), f(s')) = \cR^2(s, (a_1, a_2), s')\\
        \text{and }&\forall s, s' \in \cS: \forall (a_1, a_2) \in \cA: \cP(s, (a_1, a_2))(s') = \cP(f(s), (a_2, a_1))(f(s'))
        \end{split}
    \end{equation*}
\end{assumption}
Intuitively, given a state $s$, $f(s)$ swaps the agents: agent $1$ is in place of agent $2$ and vice versa.

\begin{lemma}[Reward Symmetry: L1]\label{lemma:reward_symmetry}
    From these two assumptions, it follows that for any states $s, s' \in \cS$, and any action $a \in A$ the following holds:
    \begin{equation*}
        \begin{split}
            \hat{\cR}^1(f(s), a, f(s')) = \hat{\cR}^2(s, a, s')\\
            \hat{\cR}^2(f(s), a, f(s')) = \hat{\cR}^1(s, a, s')
        \end{split}
    \end{equation*}
\end{lemma}

\begin{proof}
    \begin{equation*}
        \begin{split}
            \hat{\cR}^1(f(s), a, f(s')) \stackrel{A1}{=} \cR^1(f(s), (a, \cdot), f(s')) \stackrel{A2}{=} \cR^2(s, (\cdot, a), s') \stackrel{A1}{=} \hat{\cR}^2(s, a, s')\\
            \hat{\cR}^2(f(s), a, f(s')) \stackrel{A1}{=} \cR^2(f(s), (\cdot, a), f(s')) \stackrel{A2}{=} \cR^1(s, (a, \cdot), s') \stackrel{A1}{=} \hat{\cR}^1(s, a, s')
        \end{split}
    \end{equation*}
\end{proof}

During exploration, agent $1$ and $2$ follow policy $\pi_1$ and $\pi_2$, respectively. We will derive \Cref{eq:policyloss_shared,eq:valueloss_shared} for training $\pi_1$ and $V^1$ using experience collected from agent $2$. The derivation for optimisation of $\pi_2$ and $V^2$ using experience of agent $1$ can be done analogously by substituting agent indices. Note that we only derive the off-policy terms of the SEAC policy and value loss. The on-policy terms of given losses are identical to A2C~\cite{Mnih2016AsynchronousLearning}.

Agent $2$ executes action $a_2$ in state $s$. Following Assumption~\ref{assumption:symmetry}, agent $1$ needs to reinforce $\pi_1(a_2, f(s))$. Notably, in state $f(s)$, $a_1$ is sampled by $\pi_2$, so importance sampling is used to correct for this behavioural policy. 

\begin{proposition}[Actor Loss Gradient]
    \begin{equation*}
        \nabla_{\phi_1} \mathcal{L}(\phi_1) = \E_{a_2 \sim \pi_2} \left[ \frac{\pi_1(a_2 | f(s))}{\pi_2(a_2 | s)} \left( \cR^2(s, (\cdot, a_2), s') + \gamma V^1(f(s')) \right) \nabla_{\phi_1} \log \pi_1(a_2 | f(s))\right]
    \end{equation*}
\end{proposition}
\begin{proof}
    \begin{align*}
        \begin{split}
            &\nabla_{\phi_1} \mathcal{L}(\phi_1) = \E_{\substack{a_1 \sim\pi_2\\a_2 \sim \pi_1}} \left[ Q^1(f(s), a_2) \nabla_{\phi_1} \log \pi_1(a_2 | f(s))\right]\\
            &\stackrel{IS}{=} \E_{a_1, a_2 \sim \pi_2} \left[ \frac{\pi_1(a_2 | f(s))}{\pi_2(a_2 | s))} Q^1(f(s), a_2) \nabla_{\phi_1} \log \pi_1(a_2 | f(s))\right]\\
            &= \E_{a_1, a_2 \sim \pi_2} \left[ \frac{\pi_1(a_2 | f(s))}{\pi_2(a_2 | s)} \left( \cR^1(f(s), (a_2, a_1), f(s')) + \gamma V^1(f(s')) \right) \nabla_{\phi_1} \log \pi_1(a_2 | f(s))\right]\\
            &\stackrel{A1}{=} \E_{a_2 \sim \pi_2} \left[ \frac{\pi_1(a_2 | f(s))}{\pi_2(a_2 | s)} \left( \hat{\cR}^1(f(s), a_2, f(s')) + \gamma V^1(f(s')) \right) \nabla_{\phi_1} \log \pi_1(a_2 | f(s))\right]\\
            &\stackrel{L1}{=} \E_{a_2 \sim \pi_2} \left[ \frac{\pi_1(a_2 | f(s))}{\pi_2(a_2 | s)} \left( \hat{\cR}^2(s, a_2, s') + \gamma V^1(f(s')) \right) \nabla_{\phi_1} \log \pi_1(a_2 | f(s))\right]\\
            &\stackrel{A1}{=} \E_{a_2 \sim \pi_2} \left[ \frac{\pi_1(a_2 | f(s))}{\pi_2(a_2 | s)} \left( \cR^2(s, (\cdot, a_2), s') + \gamma V^1(f(s')) \right) \nabla_{\phi_1} \log \pi_1(a_2 | f(s))\right]
        \end{split}
    \end{align*}
\end{proof}

It should be noted that no gradient is back-propagated through the target $V^1(f(s'))$). In the same manner, the value loss, as shown in \Cref{eq:valueloss_shared}, can be derived as follows.

\begin{proposition}[Value Loss]
    \begin{equation*}
            \mathcal{L}(\theta_1) = \E_{a_2 \sim \pi_2} \left[\frac{\pi_1(a_2 | f(s))}{\pi_2(a_2 | s)} ||V^1(f(s)) - \left(\cR^2(s, (\cdot, a_2), s') + \gamma V^1(f(s'))\right)||^2 \right]
    \end{equation*}
\end{proposition}
\begin{proof}
    \begin{align*}
        \begin{split}
            &\mathcal{L}(\theta_1) = \E_{\substack{a_1\sim\pi_2\\a_2 \sim \pi_1}} \left[ ||V^1(f(s)) - \left( \cR^1(f(s), (a_2, a_1), f(s')) + \gamma V^1(f(s'))\right)||^2 \right]\\
            &\stackrel{IS}{=} \E_{a_1, a_2 \sim \pi_2} \left[\frac{\pi_1(a_2 | f(s))}{\pi_2(a_2 | s)} ||V^1(f(s)) - \left(\cR^1(f(s), (a_2, a_1), f(s')) + \gamma V^1(f(s'))\right)||^2 \right]\\
            &\stackrel{A1}{=} \E_{a_2 \sim \pi_2} \left[\frac{\pi_1(a_2 | f(s))}{\pi_2(a_2 | s)} ||V^1(f(s)) - \left(\hat{\cR}^1(f(s), a_2, f(s')) + \gamma V^1(f(s'))\right)||^2 \right]\\
            &\stackrel{L1}{=} \E_{a_2 \sim \pi_2} \left[\frac{\pi_1(a_2 | f(s))}{\pi_2(a_2 | s)} ||V^1(f(s)) - \left(\hat{\cR}^2(s, a_2, s') + \gamma V^1(f(s'))\right)||^2 \right]\\
            &\stackrel{A1}{=} \E_{a_2 \sim \pi_2} \left[\frac{\pi_1(a_2 | f(s))}{\pi_2(a_2 | s)} ||V^1(f(s)) - \left(\cR^2(s, (\cdot, a_2), s') + \gamma V^1(f(s'))\right)||^2 \right]
        \end{split}
    \end{align*}
\end{proof}

\section{Shared Experience Q-Learning}\label{sec:appendix_seql}
\begin{figure}[b]
    \centering
     \subfigure[t]{0.5\linewidth}
        \includegraphics[width=\linewidth]{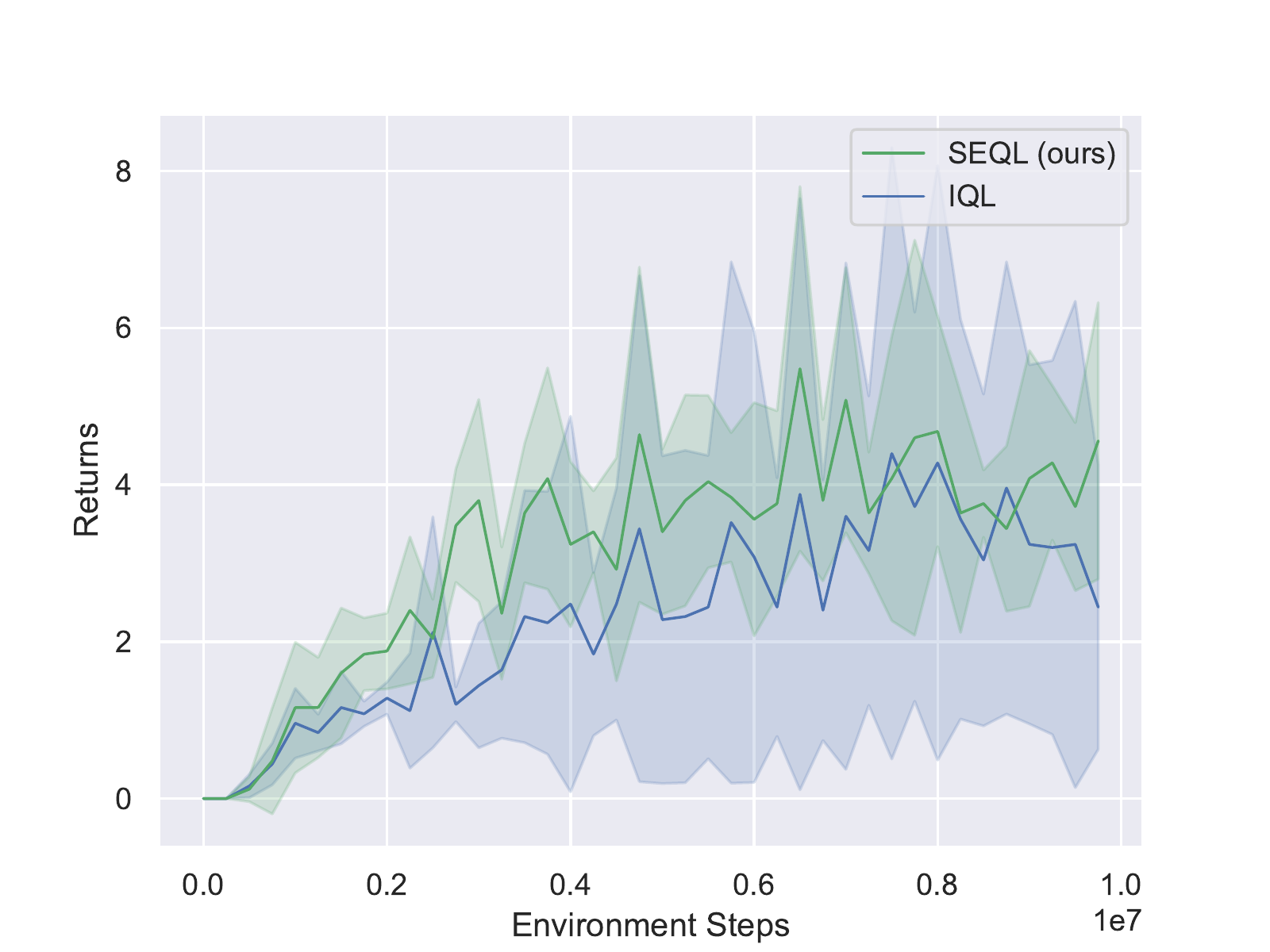}
        \caption{RWARE ($10 \times 20$), two agents}\label{fig:appendix:seql_rware-small-2ag}
    \endsubfigure\hfill
    \subfigure[t]{0.5\linewidth}
        \includegraphics[width=\linewidth]{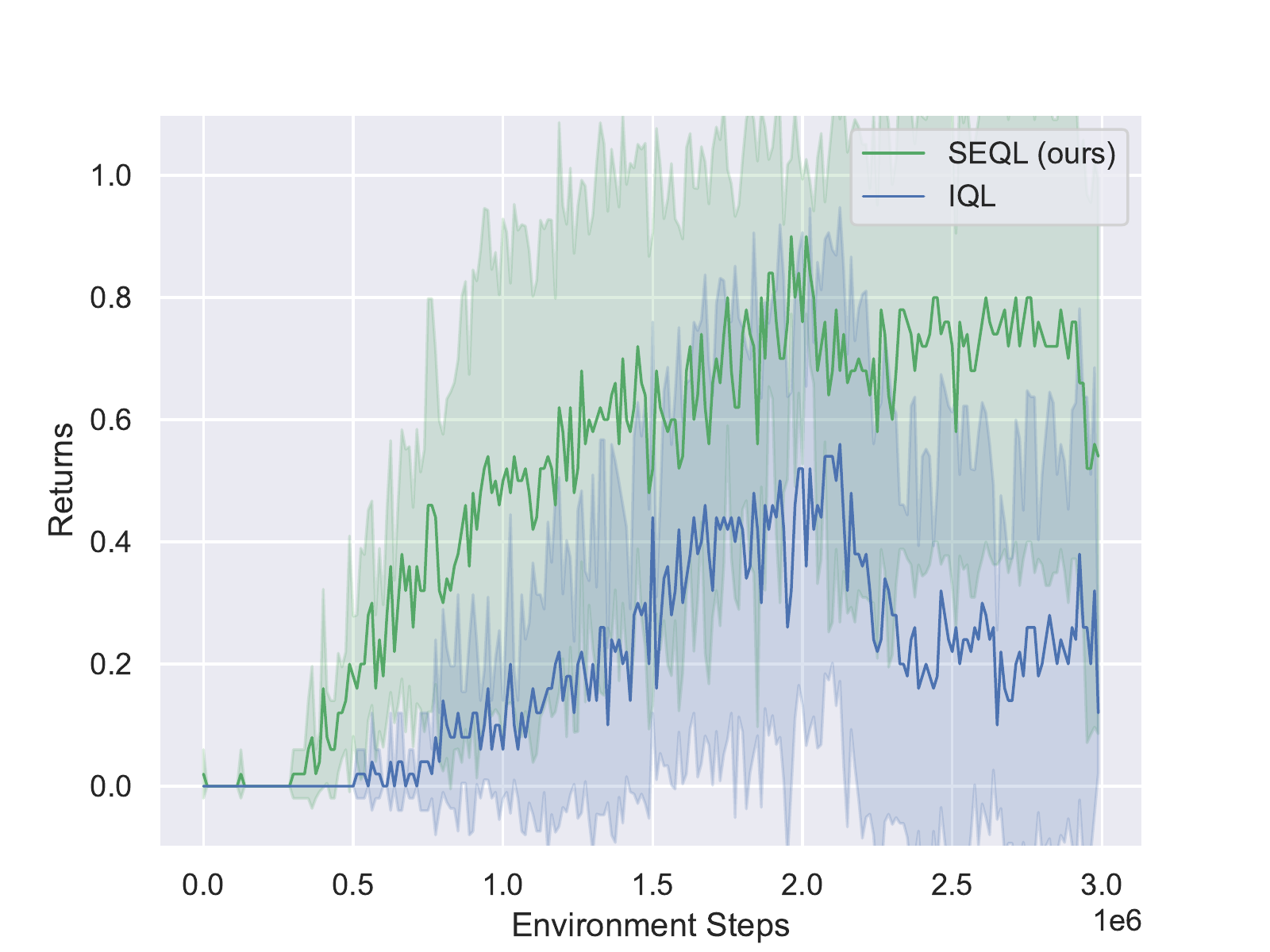}
        \caption{LBF: ($8 \times 8$), two agents, two fruits, cooperative}\label{fig:appendix:seql_lbf-8x8-2p-2f-coop}
    \endsubfigure
    \caption{Average total returns of SEQL and IQL for RWARE and LBF tasks}\label{fig:appendix:seql}
\end{figure}
\subsection{Preliminaries and Algorithm Details}
\textbf{Deep Q-Networks:} Deep Q-Networks (DQNs)~\cite{mnih2015human} are used to replace the traditional Q-tables~\cite{watkins1992q} by learning to estimate Q-values with a neural network. The algorithm uses an experience (replay) buffer $D$, which stores all experience tuples collected, circumventing the issue of time-correlated samples. Also, due to the instability created by bootstrapping, a second network with parameters $\bar\theta$ is used and updated by slowly copying the parameters of the network, $\theta$, during training. The network is trained by minimising the loss
\begin{equation}
    \begin{aligned}
    \label{eq:qloss}
        \mathcal{L}(\theta) = \frac{1}{M} \sum_{j=1}^{M}\left[(Q(s_j, a_j; \theta) - y_j)^2\right]
        \text{ with } y_j = r_j + \gamma \max_{a'}Q(s_j', a';\bar\theta)
    \end{aligned}
\end{equation}
computed over a batch of $M$ experience tuples $(s, a, r, s')$ sampled from $D$.

During each update of agent $i$, previously collected experiences are sampled from the experience replay buffer $D^i$ and used to compute and minimise the loss given in \Cref{eq:qloss}. Independently applying DQN for each agent in a MARL environment is referred to as \textit{Independent Q-Learning} (IQL)~\cite{Tan1993Multi-AgentAgents}. For such off-policy methods, sharing experience can naturally be done by sampling experience from either replay buffer $o, a, r, o'\sim D^1 \cup \ldots \cup D^N$ and using the same loss for optimisation. We refer to this variation of IQL as \textit{Shared Experience Q-Learning} (SEQL). In our experiments, we sample the same number of experience tuples $\frac{M}{N}$ from each replay buffer and the same sampled experience samples are used to optimise each agent. Hence, SEQL and IQL are optimised using exactly the same number of samples, in contrast to SEAC and IAC.


\subsection{Results}

Sharing experience in off-policy Q-Learning does improve performance, but does not show the same impact as for AC. We compare the performance of SEQL and IQL on one RWARE and LBF task to evaluate the impact of shared experience to off-policy MARL. \Cref{fig:appendix:seql} shows the average total returns of SEQL and IQL on both tasks over five seeds. In the RWARE task, sharing experience appears to reduce variance considerably despite not impacting average returns significantly. On the other hand, on the LBF task average returns increased significantly by sharing experience and at its best evaluation even exceeded average returns achieved by SEAC. However, variance of off-policy SEQL and IQL is found to be significantly larger compared to on-policy SEAC, IAC and SNAC.


\end{document}